\documentclass[12pt]{article}
\usepackage{amsmath,amssymb,amsthm}
\usepackage{color}
\usepackage{graphicx}
\usepackage{tikz}

\theoremstyle{plain}
\newtheorem{theorem}{Theorem}

\newtheorem{lemma}{Lemma}

\theoremstyle{definition}
\newtheorem{definition}{Definition}

\DeclareMathOperator*{\argmax}{arg\,max}

\begin{document}

\def\nonext{\hbox{*}}
\def\set1n{\{1,\ldots,n\}}
\def\Naturals{{\bf N}}
\def\Integers{{\bf Z}}
\def\starN{\nonext\Naturals}
\def\Rationals{{\bf Q}}
\def\Reals{{\bf R}}
\def\starR{\nonext\Reals}
\def\starRplus{\nonext \Reals_+}
\def\Rplus{\Reals_+}
\def\Rplusplus{{\Reals_{++}}}
\def\Rk{\Reals^k}
\def\Rdplus{\Reals^d_+}
\def\Rd{\Reals^d}
\def\Rm{\Reals^m}
\def\RJ{\Reals^J}
\def\RJplus{\Reals^J_+}
\def\Rmplus{\Reals^m_+}
\def\RK{\Reals^K}
\def\Rkplus{{\Reals^k_+}}
\def\Rkplusplus{{\Reals^k_{++}}}
\def\Rtwoplus{{\Reals^2_+}}
\def\Rtwoplusplus{{\Reals^2_{++}}}
\def\Deltaplusplus{{\Delta_{++}}}
\def\Deltaplus{{\Delta_{+}}}
\def\starpref{{\nonext \hspace{-.05in} \succ}}
\def\notstarpref{{\nonext \hspace{-.05in} \not\succ}}
\def\stpref{{\st \hspace{-.05in} \succ}}
\def\st{{{}^\circ}}
\def\stin{{\rm st}^{-1}}
\def\Var{{\rm Var}}
\def\ns{{\rm ns}}

\setlength{\parskip}{.15in} 


\title{Detectability, Duality, and Surplus Extraction\thanks{We thank the Editor Tilman B{\" o}rgers and referees for their very helpful comments. Financial support from the NSF under grant SES 1227707 is gratefully acknowledged. Contact: giuseppe.lopomo@duke.edu,  luca@pitt.edu, cshannon@econ.berkeley.edu}}
\date{this version: August 2021}
\author{Giuseppe Lopomo \\ Duke University \and Luca Rigotti \\ University of Pittsburgh \and Chris Shannon \\ UC Berkeley}
\maketitle
\begin{abstract}
We study surplus extraction in the general environment of McAfee and Reny (1992), and provide two alternative proofs of their main theorem. The first is an analogue of the classic argument of Cr{\' e}mer and McLean (1985, 1988), using geometric features of the set of agents' beliefs to construct a menu of contracts extracting the desired surplus. This argument, which requires a finite state space, also leads to a counterexample showing that full extraction is not possible without further significant conditions on agents' beliefs or surplus, even if the designer offers an infinite menu of contracts. The second argument uses duality and applies for an infinite state space, thus yielding the general result of McAfee and Reny (1992). Both arguments suggest  methods for studying surplus extraction in settings beyond the standard model, in which the designer or agents might have objectives other than risk neutral expected value maximization.
  
\end{abstract}

\section{Introduction}

In most settings with asymmetric information, private information generates rents for the agents who hold it. This underlies many central results, such as the inefficiency of outcomes in many mechanisms. A series of important results by Cr{\' e}mer and McLean (1985, 1988) and McAfee and Reny (1992) proved that such rents can be fragile, however, and depend crucially on the assumption that agents' private information is independent. If instead agents' information is correlated, even only to an arbitrarily small degree, then appropriately designed mechanisms can typically leverage this correlation to extract all, or virtually all, information rents. 

Cr{\' e}mer and McLean started this important strand of work. They considered the problem of surplus extraction in two particular settings, the monopolist screening problem (1985), and private values auctions (1988), each with agents whose private information is summarized by finitely many types. McAfee and Reny (1992) instead consider a very general environment, and allow for infinitely many types. McAfee and Reny (1992) show that a natural analogue of Cr{\' e}mer and McLean's convex independence condition on beliefs is necessary and sufficient for virtual surplus extraction in this class of environments, meaning that for every $\varepsilon >0$ a designer can offer a finite menu of contracts leaving each agent with no more than $\varepsilon$ surplus. 

As McAfee and Reny argue, considering the infinite type case is not merely a technical exercise in mathematical completeness, but is of central importance for understanding the explanatory power of the finite model and its ability to approximate the infinite case. Understanding the gap between virtual and full extraction is also valuable because this gap suggests the designer might face a tradeoff between revenue and information goals. When full extraction is possible the designer can typically obtain both agents' surplus and their information, while under virtual extraction the designer typically cannot obtain agents' information precisely. In McAfee and Reny (1992), for example, the designer can recover information up to at most a finite partition of the type space. Although McAfee and Reny's conditions are natural analogues of Cr{\' e}mer and McLean's for the infinite case, their result is in no sense a limit of Cr{\' e}mer and McLean's. In addition, they establish their result by a significantly different argument, from which a connection to the finite case is not clear. Cr{\' e}mer and McLean's result that full extraction holds in the finite case can be proven constructively using the separating hyperplane theorem, while McAfee and Reny's proof that virtual extraction holds in the infinite case relies on their elegant generalization of the classic Stone-Weierstrass approximation theorem. More importantly, from their work and subsequent work, it is unclear whether the gap between virtual and full extraction is due to the restriction to a finite menu of contracts, or is instead an integral consequence of allowing for infinitely many types.  

In this paper we provide two alternative proofs of the main theorem of McAfee and Reny (1992). Each provides insight into connections between the finite and infinite types models, including the extent to which the finite model approximates the infinite one, as well as suggesting novel techniques and results. Both arguments use techniques based on convex analysis and separating hyperplane theorems, as in standard approaches to full extraction with finite types. The first is a natural analogue of the argument of Cr{\' e}mer and McLean, using geometric features of the set of agents' beliefs to explicitly construct a menu of contracts extracting the desired surplus. This argument requires a finite set of states over which agents are uncertain. By highlighting the connection to the geometry of the set of beliefs, this argument also leads to a counterexample showing that full extraction is not possible without further significant conditions on agents' beliefs or surplus, even if the designer offers an infinite menu of contracts. 

The second argument is based on duality, using the characterization of extraction as the existence of a solution to a particular family of inequalities. This argument, while not constructive, applies for an infinite state space, and thus yields the general result of McAfee and Reny (1992). Both arguments  suggest methods for studying surplus extraction in other models in which agents or the designer might have objectives other than risk-neutral value maximization. 

Our use of duality in the surplus extraction problem is inspired by and builds on the work of Rahman (2012), which introduced duality arguments to study surplus extraction in arbitrary type spaces. Rahman (2012) argues that full surplus extraction holds under an analogue of convex independence in a setting with general type spaces allowing the designer to offer an infinite menu of contracts.
We show by example that full extraction can fail under the assumptions of McAfee and Reny (1992), even allowing for an infinite menu of contracts. The two proofs we provide shed some light on why virtual extraction holds while full extraction can fail, and the extent to which duality arguments can be used to study surplus extraction in environments more general than the standard model. Other recent papers have also emphasized the importance of duality for different mechanism design questions. This includes work on multidimensional screening using optimal transport methods by Daskalakis, Deckelbaum, and Tzamos (2017), and work on optimal auction design under robustness concerns by Carroll and Segal  (2018) and Bergemann, Brooks, and Morris (2017a, b).   

The paper proceeds as follows. In section 2 we give some preliminary definitions and results from convex analysis that will be used throughout the paper. In section 3 we set up the basic model and definitions, including notions of surplus extraction. In section 4 we give a necessary condition for full extraction, and use this to derive a counterexample showing that full extraction can fail under the assumptions of McAfee and Reny (1992), even if the designer can offer an infinite menu of contracts. We then consider the case in which the state space is finite, and give a constructive proof that virtual extraction holds in this setting. In section 5 we consider the  case in which the state space can be infinite, and provide a proof of the general virtual extraction result based on duality. In section 6 we give an example illustrating how these methods can be used to study surplus extraction in models beyond the standard case. Additional results, some of which might be of independent interest, are collected in the appendix.   
 
\section{Preliminaries}

We recall and collect here some preliminary definitions and results from convex analysis.

\begin{definition}
Let $X$ be a topological vector space and $C\subseteq X$ be a convex set. An \emph{extreme point} of $C$ is a point $x\in C$ with the property that if $x=\alpha y + (1-\alpha) z $ for some $y,z\in C$ and some $\alpha \in [0,1]$, then $x=y$ or $x=z$. 

An \emph{exposed point} of $C$ is a point $x\in C$ such that there is some continuous real linear functional $f$ on $X$ such that $f(y) < f(x)$ for all $y\in C$ with $y\not= x$. 
\end{definition}

\medskip

\noindent {\bf Note:} Every exposed point is an extreme point, but the converse does not hold. That is, extreme points need not be exposed. If $X$ is locally convex and $C\subseteq X$ is compact and convex and has only finitely many extreme points, however, then every extreme point of $C$ is exposed. 

\medskip

\begin{definition}
A nonempty subset $F\subseteq C$ of a convex set $C$ is a \emph{face} of $C$ if $F$ is convex and whenever $x,y\in C$ and $\alpha x + (1-\alpha) y\in F$ for some $\alpha \in (0,1)$, then $x,y\in F$. A face $F$ of $C$ is a \emph{proper face} if it is a proper subset of $C$. 
\end{definition}

\medskip

\noindent {\bf Note:} Every face is a convex set by definition. Every extreme point is a (singleton) face, but a face can have more than one element in general. 

\medskip

\begin{definition}
Let $X$ be a topological vector space and $C\subseteq X$ be a convex set. A nonempty subset $E\subseteq C$ is an \emph{exposed set} of $C$ if there is a continuous real linear functional $f$ on $X$ such that $f(y) \leq f(x)$ for all $x,y\in C$ with $x\in E$, and $f(y) < f(x)$ if $y\not\in E$. 
\end{definition}

\medskip

\noindent {\bf Note:} An exposed set is a face, but a face need not be exposed. 

\bigskip

\begin{definition}
Let $C\subseteq {\bf R}^n$ be a convex set and let $W\subseteq {\bf R}^n$ be the unique affine subspace of ${\bf R}^n$ such that $C\subseteq W$ and such that $C$ has a nonempty relative interior in $W$. The \emph{dimension} of $C$, denoted $\mbox{ dim } C$, is the dimension of $W$. 
\end{definition}

\bigskip

We record a useful result that connects these concepts next. 

\begin{theorem}
Let $C\subseteq {\bf R}^n$ be a compact convex set, and $F\subseteq C$ be a proper face of $C$. Then $\mbox{ dim } F < \mbox{ dim } C$. 
\end{theorem}

For example, see Simon (2011, Proposition 8.10).

Finally, note that if $C\subseteq {\bf R}^n$ is a compact, convex set with $\mbox{dim } C = 1$, then $C$ has finitely many extreme points, 
and thus all extreme points of $C$ are exposed.  

\section{Set-up and Extraction Notions}

In this section we lay out the basic set-up and notation used throughout the paper, and give the definitions of surplus extraction underlying the main results. 

We use the following standard notation. For a compact metric space $B$, $C(B)$ is the space of continuous real-valued functions on $B$, and ${\cal M}(B)$ is the space of finite signed Borel measures on $B$. Similarly, $\Delta(B)$ is the space of Borel probability measures on $B$. 

For $x\in C(B)$ and $\eta \in {\cal M}(B)$, we write $x\cdot \eta = \eta \cdot x$ for the bilinear form $\langle x, \eta \rangle = \langle \eta, x\rangle$, that is, 
\[
\eta \cdot x = x\cdot \eta = \langle x, \eta \rangle = \langle \eta, x \rangle = \int x(b) \ \eta(db)
\]

We follow McAfee and Reny (1992) in giving a reduced form description of the surplus extraction problem. In a prior unmodeled stage, agents play a game that leaves them with some information rents as a function of their private information. Private information is summarized by the type $t\in T$, where $T$ denotes the set of possible types. Unless specified otherwise, we let $T = [0,1]$ be the set of types (for all of the main results it suffices that $T$ is a compact metric space). The current stage also has an exogenous source of uncertainty, summarized by a set of states $S$, on which contract payments can depend. We follow McAfee and Reny (1992) in allowing $S$ to be arbitrary, and assume throughout that $S$ is a compact metric space. For some applications, it is natural to take $S=T^n$ for some $n$, as in the original results of Cr{\' e}mer and McLean (1988) on auctions, or other settings with multiple agents in which this uncertainty is over the profile of agents' types. In these applications the cardinality of $S$ is greater than or equal to the cardinality of $T$, and in particular $S$ is infinite whenever $T$ is infinite. In other applications this public information is unrelated to agents' private information, so $S$ and $T$ are independent. Such applications include contracting with public ex post information, as in the work orginating with Riordan and Sappington (1988).\footnote{We thank a referee for suggesting this example.} In these settings $T$ might be infinite while $S$ is finite; we consider this possibility in particular in section 4.  

To each type $t\in T$ is then associated a value $v(t)\in {\bf R}$ and beliefs $\pi(t) \in \Delta (S)$. We typically interpret the value $v(t)$ as rents from the prior stage, but this could also represent any other revenue target of the designer consistent with individual rationality.\footnote{We thank Tilman B{\" o}rgers for suggesting this interpretation. See also the discussion at the end of this section.} Throughout we maintain the assumption that $v:T\to {\bf R}$ is continuous and that $\pi:T\to \Delta(S)$ is norm continuous.  

Let 
\[
C:= \mbox{co} \{ \pi(t) \in \Delta (S): t\in T \}
\]
where for a subset $A$ of a topological vector space $X$, $\mbox{co} (A)$ denotes the convex hull of $A$, and $\overline{\mbox{co}} (A)$ denotes the closed convex hull of $A$. 

Following our general notation, let $\Delta(T)$ be the set of Borel probability measures on $T$. For $t\in T$, $\delta_t\in \Delta(T)$ denotes the Dirac measure concentrated on $t$. For $x\in C(T\times S)$ and $t\in T$, we write $x(t) \in C(S)$ for the function such that $x(t) (r) = x(t,r)$ for each $r\in S$. Throughout we use $s$ and $t$ as generic elements of $T$.

Next we give definitions for the main notions of surplus extraction in this setting, full extraction and virtual extraction. Both reflect the idea that the designer offers agents a menu of stochastic contracts from which they choose, based on minimizing their expected costs. Exploiting correlation between types and beliefs might allow the designer to construct such a menu that leaves every agent with zero expected surplus, in the case of full extraction, or no more than $\varepsilon$ expected surplus for any $\varepsilon >0$, in the case of virtual extraction.

\bigskip
\begin{definition}
\emph{Full extraction} holds if, for each given $v:T\to {\bf R}$, there exists a collection $\{c(t) \in C(S) : t\in T\}$ such that for each $t\in T$:
\[ 
v(t) - \pi(t)\cdot c(t) = 0
\]
and
\[
v(t) -\pi(t)\cdot c(s) \leq 0 \ \ \ \forall s\not= t
\]
\emph{Virtual extraction} holds if, for each given $v:T\to {\bf R}$, for each $\varepsilon >0$ there exists a collection $\{c_\varepsilon(t)\in C(S): t\in T\}$ such that for each $t\in T$:
\[
0\leq v(t) - \pi(t)\cdot c_\varepsilon(t) \leq \varepsilon
\]
and
\[
v(t) -\pi(t)\cdot c_\varepsilon(s) \leq \varepsilon \ \ \ \forall s\not= t
\]
\end{definition}

\bigskip

Equivalently, if $T$ is endowed with some relevant measure, each definition above could require that these conditions hold for almost every $t, s \in T$ instead of for all $t, s\in T$. 

As defined, full extraction or virtual extraction might require the designer to offer an infinite menu of contracts when $T$ is infinite. In the case of virtual extraction, such a menu need not have an expected cost minimizing element for all agents. By allowing for an infinite menu of contracts, this might also appear to be a weaker notion of virtual extraction than considered by McAfee and Reny (1992), which instead shows that for each $\varepsilon >0$, there is a finite menu $\{ c_1, \ldots ,c_n\}$ such that for each $t\in T$, 
\[
0\leq \max_{j=1,\ldots ,n} \lbrace v(t) -\pi(t) \cdot c_j \rbrace \leq \varepsilon 
\]
We note, however, that whenever virtual extraction holds (using the definition above), then it is always possible to find a finite menu of contracts that would achieve the same bounds on surplus. We record this observation and its proof below.

\begin{theorem}
If virtual extraction holds, then virtual extraction can be achieved with a finite menu of contracts. That is, given $v:T\to {\bf R}$, for each $\varepsilon >0$ there exists a finite menu $\{ c_1,\ldots ,c_n\}\subseteq C(S)$ such that for each $t\in T$, 
\[
0\leq \max_{j=1,\ldots ,n} \lbrace v(t) -\pi(t) \cdot c_j \rbrace \leq \varepsilon 
\]
\end{theorem}
\begin{proof}
Let $v:T\to {\bf R}$ be given and fix $\varepsilon >0$. Choose $\{c_\varepsilon(t) \in C(S) : t\in T\}$ such that for each $t\in T$:
\[
0\leq v(t) - \pi(t) \cdot c_\varepsilon (t) \leq \varepsilon 
\]
and
\[
0\leq \sup_{s\in T} \lbrace v(t) - \pi(t) \cdot c_\varepsilon (s) \rbrace \leq \varepsilon 
\]
Then for each $t\in T$ there exists $\delta_t>0$ such that $s\in B_{\delta_t}(t) \Rightarrow$
\[
-\frac{\varepsilon}{2} \leq v(s) - \pi(s)\cdot c_\varepsilon(t) \leq \varepsilon 
\]
where $B_{\delta_t}(t) = \{ s\in T: \Vert s-t \Vert < \delta_t\}$. Since $T$ is compact and $\{ B_{\delta_t}(t) : t\in T\}$ is an open cover of $T$, there exists $t_1,\ldots ,t_n$ such that $T\subseteq \cup_i B_{\delta_{t_i}}(t_i)$. Then for each $i$, set $c_i = c_\varepsilon(t_i) - \varepsilon$. For each $t\in T$ there exists $i$ such that $t\in B_{\delta_{t_i}}(t_i)$. Thus 
\[
0\leq v(t) -\pi(t) \cdot c_i \leq 2\varepsilon 
\]
and for each $j=1,\ldots ,n$, 
\[
v(t) -\pi(t) \cdot c_j \leq 2\varepsilon 
\]
Thus
\[
0\leq \max_{j=1,\ldots ,n} \lbrace v(t) -\pi(t) \cdot c_j \rbrace \leq 2\varepsilon 
\]
The result follows.
\end{proof}

\bigskip

As with the original formulation of McAfee and Reny (1992), these notions of full and virtual extraction do not explicitly address incentive compatibility. When full extraction holds, incentive compatibility will follow. Incentive compatibility does   not immediately follow from the definition of virtual extraction, however; this would require that the menu of contracts satisfies the additional incentive constraints $v(t) -\pi(t)\cdot c(s) \leq v(t) - \pi(t) \cdot c(t)$ for all $s,t\in T$. But if virtual extraction holds, then there is always some menu achieving virtual extraction that is also incentive compatible, by virtue of Theorem 2. If virtual extraction holds, then it can be achieved with some finite menu of contracts $\{c_1, \ldots ,c_n\}$. Choosing $c(t) \in \argmax \{ v(t) - \pi(t) \cdot c_i : i=1,\ldots ,n\}$ for each $t\in T$ then yields a menu that achieves virtual extraction and is also incentive compatible. More generally, an infinite menu of contracts might allow for virtual extraction while also satisfying incentive compatibility, but need not. Our results do not address this question directly, although they suggest that probabilistic independence is not sufficient to guarantee incentive compatibility for every infinite menu.

Following Cr{\' e}mer and McLean (1988) and McAfee and Reny (1992), we consider conditions on beliefs under which full extraction or virtual extraction holds. McAfee and Reny (1992) show that virtual extraction is possible whenever beliefs satisfy the following condition.

\bigskip

\begin{definition}
Types satisfy {\it probabilistic independence} if for all $t\in T$:
\[
\pi(t) = \int \pi(s) \mu(ds) \mbox{ for some } \mu\in \Delta(T) \Rightarrow \mu = \delta_t
\]
\end{definition}

\noindent {\bf Note:} If types satisfy probabilistic independence, then $\pi(t)$ is an extreme point of $C$ for each $t\in T$.

\noindent {\bf Note: } If $T$ is finite, then probabilistic independence reduces to the standard {\it convex independence} condition of Cr{\' e}mer-McLean, that is 
\[
\pi(t) = \sum_{s\in T} \mu_s \pi(s) \mbox{ for some } \mu\in \Delta(T) \Rightarrow \mu_t=1
\]

Cr{\' e}mer and McLean (1988) show that when $T$ and $S$ are finite, then full extraction holds whenever beliefs satisfy convex independence. This is no longer true when $T$ is infinite, as the example in the next section illustrates. We sketch a standard argument for the classic result when $T$ and $S$ are both finite next, to motivate the main ideas we develop in the following sections. 

Suppose $T$ and $S$ are finite, and types satisfy convex independence. Fix a type $t\in T$. Since types satisfy convex independence, $\pi(t) \not\in \overline{\mbox{co}} \{ \pi(s) : s\in T, s\not= t\}$. Thus there exists $z(t) \in {\bf R}^S$ such that 
\[
\pi(t) \cdot z(t) = 0
\]
and
\[
\pi(s) \cdot z(t) >0 \ \ \forall s\in T, \ s\not= t
\]
That is, $\pi(t)$ is an exposed point of $\{ \pi(s): s\in T\}$, and of $C = \mbox{co} \{ \pi(s): s\in T\}$. Alternatively, by convex independence, $\{ \pi(s): s \in T\}$ is the set of extreme points of $C$; since $T$ is finite this set is finite, so each element $\pi(t)$ must also be an exposed point of $C$. Now consider a contract of the form $c(t) = v(t) + \alpha(t) z(t)$ where $\alpha(t) \in {\bf R}_+$, which requires the constant payment $v(t)$ and a stochastic payment that is a scaled version of $z(t)$ (throughout we use a constant $r\in {\bf R}$ interchangeably with the function $r{\bf 1}(S)$, where ${\bf 1}(S)$ denotes the identity on $S$). For type $t$, this contract has expected cost $v(t)$, as 
\[
\pi(t) \cdot c(t) = v(t) + \alpha(t) (\pi(t) \cdot z(t)) = v(t)
\]
while for types $s\not= t$, the expected cost is
\[
\pi(s) \cdot c(t) = v(t) + \alpha(t) (\pi(s)\cdot z(t))
\]
Since $\pi(s)\cdot z(t)>0$ for all types $s\not= t$, the designer can take advantage of this difference in beliefs to set $\alpha(t)$ sufficiently large to make the resulting contract unattractive to all types $s\not= t$ while keeping the expected cost constant for type $t$. To that end, set $\alpha(t)>0$ sufficiently large so that
\[
\alpha(t) > \max\limits_{s\not= t} \frac{v(s)-v(t)}{\pi(s)\cdot z(t)}
\]
Note that since $\pi(s)\cdot z(t)>0$ for all $s\not= t$ and $T$ is finite, the term on the right above is well-defined (the ratio is bounded), and thus $\alpha(t)$ is well-defined. Then for type $t$,
\[
v(t) - \pi(t) \cdot c(t) = v(t) - v(t) = 0
\]
while for types $s\not= t$, 
\[
v(s) - \pi(s) \cdot c(t) = v(s) - v(t) - \alpha(t) ( \pi(s) \cdot z(t)) <0
\]
by choice of $\alpha(t)$. Repeating this construction for each type $t\in T$ yields a menu $\{  c(t) : t\in T\}$ that achieves full extraction. Finally, if in addition the cardinality of $S$ is greater than or equal to the cardinality of $T$, then convex independence is satisfied for almost all elements of $\Delta(S)^T$.

Under a stronger condition on beliefs, full extraction can be achieved while satisfying ex post incentive compatibility. In this simple reduced form setting, ex post incentive compatibility is equivalent to the menu $\{ c(t) : t\in T\}$ consisting of a single contract, that is, $c(t) = c$ for all $t\in T$. For this stronger form of full extraction, suppose that $\{ \pi(t) : t\in T\}$ is linearly independent. Thus for each $t\in T$, $\pi(t) \not\in \mbox{ span } \{ \pi(s): s\not= t\} = \{ \sum_{s\not= t} \alpha_s \pi(s) : \alpha_s \in {\bf R} \ \forall x\not= t\}$.  Fix $t\in T$. Since $\pi(t) \not\in \mbox{ span } \{ \pi(s): s\not= t\}$, there exists $z(t) \in {\bf R}^S$ such that 
\[
\pi(t) \cdot z(t) >0
\]
and
\[
\pi(s) \cdot z(t) = 0  \ \ \forall s\not= t
\]
Then set
\[
c(t) = \frac{v(t)}{\pi(t)\cdot z(t)} z(t)
\]
and note that $\pi(t) \cdot c(t) = v(t)$ while $\pi(s) \cdot c(t) = 0$ for all $s\not= t$. Repeat this argument for each $t\in T$.  Define
\[
c = \sum_{t\in T} c(t)
\]
Then for each $t\in T$, $\pi(t) \cdot c = v(t)$. Thus full extraction holds under ex post incentive compatibility. This is the analogue in this reduced form setting of the classic results of Cr{\' e}mer and McLean (1988) for dominant strategy mechanisms. 
 
We close this section by illustrating how this reduced form formulation can incorporate standard problems, including those with multiple agents, by considering quasilinear mechanism design problems with $n$ agents. Each agent $i=1,\ldots ,n$ has a type $t_i \in T_i$ and belief $\pi_i(t_i) \in \Delta(T_{-i})$, where $T_{-i} = \prod_{j\not= i} T_j$. To simplify this discussion and the connections with the standard formulation as in Cr{\'e}mer and McLean (1988), here we take $T_i$ to be finite for each $i$. A direct mechanism $(q,p)$ consists of an allocation function $q:\prod_{i=1}^n T_i \to {\bf R}^n$ and a payment function $p:\prod_{i=1}^n T_i \to {\bf R}^n$. Given $i$ and $t_i\in T_i$, write $q_i(t_i), p_i(t_i) \in {\bf R}^{\vert T_{-i}\vert}$ for the functions such that $q_i(t_i)(t_{-i}) =  q_i(t_i, t_{-i})$  and $p_i(t_i)(t_{-i}) = p_i(t_i, t_{-i})$ for each $t_{-i} \in T_{-i}$. Here $\vert T_{-i}\vert$ denotes the cardinality of $T_{-i}$. 

For each agent $i$, given $t_{-i}$, the payoff from reporting $t_i'$ in this mechanism when true type is $t_i$ is 
\[
t_i q_i(t_i', t_{-i}) - p_i(t_i', t_{-i})
\] 
Define $v_i:T_i \to {\bf R}$ by
\[
v_i(t_i) = \max_{t_i'\in T_i}  \pi_i(t_i) \cdot ( t_i q_i(t_i') -  p_i(t_i') )
\]
This gives the information rent of agent $i$ in this mechanism.

Suppose for each agent $i$, beliefs $\{ \pi_i(t_i): t_i\in T_i\}$ satisfy convex independence. Then following the argument above, full extraction of $v_i$ is possible for each $i$. This implies, for example, if $(q,p)$ is Bayesian incentive compatible, then there is a Bayesian incentive compatible mechanism $(q, p')$ with the same allocation function $q$ that extracts all the surplus from the original mechanism, that is, such that $\pi_i(t_i) \cdot p_i'(t_i) = v_i(t_i)$ for all $t_i\in T_i$ and for all $i$. 

In this case, a stronger result is also possible, as in B{\" o}rgers (2015). Given any direct mechanism $(q,p)$, there is a direct mechanism $(q,p')$ that is Bayesian incentive compatible, has the same allocation rule $q$, and for which interim expected payments in $p'$ are equivalent to those in $p$, that is, $\pi_i(t_i) \cdot p_i(t_i) = \pi_i(t_i) \cdot p_i'(t_i)$ for each $t_i\in T_i$ and each $i$. While the surplus function $v_i$ defined above is too coarse to derive this result, a straightforward extension of the arguments above can account for incentives to misreport in the original mechanism. To that end, for each $i$ define $u_i:T_i\times T_i \to {\bf R}$ by
\[
u_i(t_i, t_i') =  \pi_i(t_i) \cdot ( t_i q_i(t_i') -   p_i(t_i') )
\]
Modifying the arguments above yields for each $i$ a menu $\{z_i(t_i)\in {\bf R}^{\vert T_{-i}\vert} : t_i\in T_i\}$ such that $\pi_i(t_i) \cdot z_i(t_i) = 0$ while $\pi_i(t_i) \cdot z_i(t_i') > u_i(t_i, t_i') - u_i(t_i, t_i)$ for any $t_i'\not= t_i$. Then using the payment functions $p_i'(t_i) = p_i(t_i) + z_i(t_i)$ for each $t_i\in T_i$ and for each $i$, the mechanism $(q,p')$ is Bayesian incentive compatible while $\pi_i(t_i) \cdot p_i'(t_i) = \pi_i(t_i) \cdot p_i(t_i)$ for each $t_i\in T_i$ and for each $i$. 

Similarly, under the stronger condition on beliefs that $\{\pi_i(t_i) : t_i\in T_i\}$ is linearly independent for each $i$, full extraction can be achieved with dominant strategy incentive compatibility. In this case, if $(q,p)$ is dominant strategy incentive compatible, there is a dominant strategy incentive compatible mechanism $(q,p')$ with same allocation rule $q$ that extracts all surplus from $(q,p)$.

\section{Detectability and Finite State Space}

In this section, we connect the extraction problem to the underlying convex geometry, as in the setting with finitely many types. This lets us establish a more direct connection between the extraction problem with finitely many types and with infinitely many types. From this connection we identify a simple necessary condition for full extraction, based on the geometry of the set of beliefs $\{ \pi(t): t\in T\}$. This necessary condition sheds light on why probabilistic independence is no longer sufficient for full extraction with infinitely many types, and allows us to give an example to illustrate this breakdown. The example also sheds light on why virtual extraction holds nonetheless, and on the nature of types to whom some surplus might need to be left. Finally, these observations lead to a constructive proof that virtual extraction holds under probabilistic independence, analogous to classic arguments in the finite type case, in which we will explicitly construct a menu of contracts to achieve extraction of all but at most $\varepsilon$ surplus for each $\varepsilon > 0$. This argument requires the state space $S$ to be finite. When $S$ is finite, with abuse of notation, we use the symbol $S$ interchangeably for the state space and its cardinality in this section. 

We start with the observation that if full extraction holds, then for each $t\in T$, $\pi(t)$ must be an exposed point of $\{ \pi(s): s\in T\}$ and of $C= \mbox{co} \{ \pi(s) : s\in T \}$.\footnote{This can be viewed as a strengthening of Proposition 2 in Heifetz and Neeman (2006), which shows that full surplus extraction requires the ``beliefs-determine-preferences'' condition that distinct types must hold distinct beliefs. We thank a referee for this observation.}  

\begin{theorem}
If full extraction holds, then for each $t\in T$, $\pi(t)$ is an exposed point of $\{ \pi(s): s\in T\}$ and of $C= \mbox{co} \{ \pi(s) : s\in T \}$. 
\end{theorem}
\begin{proof}
Fix $t\in T$. Choose $v:T\to {\bf R}$ such that $v(s) > v(t)$ for all $s\not= t$. Since full extraction holds, there exists a collection $\{ c(s)\in C(S): s\in T\}$ such that 
\[
v(t) - \pi(t) \cdot c(t) = 0
\]
and
\[
v(s) - \pi(s) \cdot c(t) \leq 0 \ \ \ \forall s\not= t
\]
Write $c(t) = v(t) + z(t)$, where $z(t) \in C(S)$ and with abuse of notation $v(t) = v(t) {\bf 1}_S$. Since $\pi(t) \cdot c(t) = v(t)$, 
\[
\pi(t) \cdot z(t) =0
\]
Then for $s\not= t$, 
\[
\pi(s) \cdot c(t) = v(t) + \pi(s) \cdot z(t) \geq v(s)
\]
Thus
\[
\pi(s) \cdot z(t) \geq v(s)-v(t) >0 \ \ \ \forall s\not= t
\]
That is, $\pi(t)$ is an exposed point of $\{ \pi(s): s\in T\}$. 

Now take $\pi\in C= \mbox{co} \{ \pi(s) : s\in T \} \subseteq \Delta(S)$ with $\pi\not= \pi(t)$. Then by definition, there exists $\{ t_1,\ldots , t_n\} \subseteq T$ and $\alpha_1, \ldots ,\alpha_n >0$ with $\sum_i \alpha_i =1$ such that $\pi = \sum_i \alpha_i \pi(t_i)$. Since $\pi \not= \pi(t)$, there exists $i$ such that $t_i \not= t$. Then
\[
\pi \cdot z(t) = \sum_i \alpha_i \pi(t_i)\cdot z(t) >0
\]
since $\exists i $ such that $t_i \not= t$ and $\pi(t_i)\cdot z(t) >0$ for all such $t_i$. 

Thus $\pi(t)$ is an exposed point of $C$. Since $t\in T$ was arbitrary, $\pi(t)$ is an exposed point of  $\{ \pi(s): s\in T\}$, and of $C= \mbox{co} \{ \pi(s) : s\in T \}$ for each $t\in T$. 
\end{proof}

\bigskip

This observation motivates the following definition, reframing the necessary condition above in terms of agents' beliefs.

\begin{definition}
Type $t\in T$ is {\it detectable} if $\exists \ z\in C(S)$ such that 
\begin{eqnarray*}
\pi(t) \cdot z &=& 0\\
\pi(s) \cdot z &>& 0 \ \forall s\not= t
\end{eqnarray*}
\end{definition}

\bigskip

\noindent{\bf Note: } A type $t$ is detectable if and only if $\pi(t)$ is an exposed point of $C$. Using this terminology, we can then restate Theorem 3 as follows: full extraction requires that all types are detectable.

\bigskip

Next we give an example to illustrate the failure of full extraction when $T$ is infinite. The example illustrates why full extraction can fail and also suggests why virtual extraction can hold despite the failure of full extraction. In the example, probabilistic independence is satisfied, but some types are not detectable.  

\begin{figure}

\centering 

\begin{tikzpicture}[scale = 0.6]

\coordinate [label=right: $\pi(0)$] (a) at (4,1);

\coordinate [label=right: $\pi(1)$] (b) at (4,-1);

\fill (4,1) circle (2pt);

\fill (4,-1) circle (2pt);

\draw  (4,1) .. controls (4,4) and (0,4) .. (0,0);

\draw  (0,0) .. controls (0,-4) and (4,-4) .. (4,-1);

\end{tikzpicture}
\caption{Graph of $\{ \pi(t): t\in T\}$ }
\end{figure}

\noindent {\bf Example 1: } Let $\pi:T\to \Delta(S)$ be as in Figure 1, and suppose $v(t)>v(0)$ for all $t\not= 0$. Note that for every $t\not= 0,1$, $\pi(t)$ is an exposed point of $C= \mbox{co} \{ \pi(s) : s\in T \}$, while $\pi(0)$ and $\pi(1)$ are extreme points of $C$ that are not exposed. That is, $t$ is detectable in $T$ for all $t\not= 0,1$, while types $t=0$ and $t=1$ are not detectable in $T$. See Figures 2 and 3.

To see that these beliefs satisfy probabilistic independence, suppose 
\[
\pi(t) = \int \pi(s)\ \mu(ds) \ \ \ \ \mbox{ for some } t\in T \mbox{ and } \mu\in \Delta(T)
\]
First suppose $t\not\in \{ 0,1 \}$. Then since $\pi(t)$ is an exposed point of $C$, there exists $z(t) \in {\bf R}^S$ such that $\pi(t) \cdot z(t) =0$ 
and $\pi(s) \cdot z(t) >0$ \ $\forall s\not= t$. Then
\[
0 = \pi(t) \cdot z(t) = \int \pi(s)\cdot z(t) \ \mu(ds)
\]
Since $\pi(s)\cdot z(t) >0$ for all $s\not= t$ and $\pi(t) \cdot z(t) =0$, this implies $\mu = \delta_t$. 

Now suppose $t\in \{ 0,1\}$. In this case, there exists $z(t) \in {\bf R}^S$ such that $\pi(t) \cdot z(t) = \pi(0)\cdot z(t) = \pi(1) \cdot z(t) =0$
and $\pi(s) \cdot z(t) >0$ \ $\forall s\not= 0,1$. Then as above, 
\[
0=\pi(t)\cdot z(t) = \int \pi(s)\cdot z(t) \ \mu(ds)
\]
Since $\pi(s)\cdot z(t) >0$ for all $s\not= 0,1$ and $\pi(0)\cdot z(t) = \pi(1)\cdot z(t) =0$, this implies $\mbox{ supp } \mu \subseteq \{ 0,1\}$. Thus $\pi(t) = \alpha \pi(0)+ (1-\alpha) \pi(1)$ for some $\alpha \in [0,1]$. But since both $\pi(0)$ and $\pi(1)$ are extreme points of $C$, this implies $\mu = \delta_t$. 

Thus beliefs satisfy probabilistic independence. Full extraction does not hold, however, by Theorem 3, since $\pi(0)$ and $\pi(1)$ are not exposed points in $C$, that is, types $t=0$ and $t=1$ are not detectable in $T$.\footnote{Similar examples can be constructed for any compact metric space $S$ with at least three elements. The key condition necessary for such an example is that the set has infinitely many extreme points, and thus that $T$ is infinite.}

\begin{figure}
\centering
\begin{tikzpicture}[scale = 0.6]

\coordinate [label=right: $\pi(0)$] (a) at (4,1);

\coordinate [label=right: $\pi(1)$] (b) at (4,-1);

\fill (4,1) circle (1pt);

\fill (4,-1) circle (1pt);

\draw  [fill = gray!20] (4,1) .. controls (4,4) and (0,4) .. (0,0);

\draw [fill=gray!20] (0,0) .. controls (0,-4) and (4,-4) .. (4,-1);

\draw [gray!20, fill=gray!20] (0,0) -- (4,1) -- (4,-1) -- (0,0) ;

\draw (4,1) -- (4,-1);

\fill (4,1) circle (2pt);

\fill (4,-1) circle (2pt);

\end{tikzpicture}

\caption{Graph of $C = \mbox{co} \{ \pi(t): t\in T\}$; $\pi(0)$ and $\pi(1)$ are extreme points of $C$ that are not exposed.}
\end{figure}

In this case, it is not difficult to see directly why full extraction fails. As above, suppose $v(t) >v(0)$ for all $t\not= 0$. Notice in particular this implies $v(1)>v(0)$, so the surplus of type 1 is greater than the surplus of type 0. Now if full extraction were possible, there must exist a contract $c(0) = v(0) + z(0)$ for some $z(0) \in {\bf R}^S$ such that
\[
v(0) = \pi(0)\cdot c(0) = v(0) + \pi(0)\cdot z(0) \iff \pi(0) \cdot z(0) =0
\]
and 
\[
v(s) \leq \pi(s) \cdot c(0) = v(0) + \pi(s)\cdot z(0) \iff \pi(s) \cdot z(0) \geq v(s)-v(0) >0 \ \ \ \forall s\not= 0
\]

\begin{figure}
\centering
\begin{tikzpicture}[scale = 0.6]

\coordinate [label=right: $\pi(0)$] (a) at (4,1);

\coordinate [label=right: $\pi(1)$] (b) at (4,-1);

\fill (4,1) circle (1pt);

\fill (4,-1) circle (1pt);

\draw  [fill = gray!20] (4,1) .. controls (4,4) and (0,4) .. (0,0);

\draw [fill=gray!20] (0,0) .. controls (0,-4) and (4,-4) .. (4,-1);

\draw [gray!20, fill=gray!20] (0,0) -- (4,1) -- (4,-1) -- (0,0) ;

\draw (4,1) -- (4,-1);

\draw (4,3) -- (4, -3);

\fill (4,1) circle (2pt);

\fill (4,-1) circle (2pt);

\end{tikzpicture}

\caption{Full extraction is impossible when $v(t)>v(0)$ for all $t\not= 0$. }
\end{figure}

Here, however, if $\pi(0) \cdot z=0$ and $\pi(s) \cdot z \geq 0$ for all $s\in T$, it must be the case that $\pi(1) \cdot z=0$ as well (see Figure 3). But then the contract $c(0)$ must leave type 1 with strictly positive surplus, as 
\[
v(1) - \pi(1) \cdot c(0) = v(1) - v(0) -\pi(1)\cdot z(0) = v(1)-v(0) >0
\] 
So any contract that extracts full surplus from type 0 and does not provide surplus to other types $t\in (0,1)$ must leave strictly positive surplus for type 1. \hfill$\diamondsuit$

\bigskip

Although full extraction is not possible in this example, virtual extraction is. Figures 4 and 5 illustrate the idea, and also illustrate that it might be necessary to give small rents to some types close to a type that is not detectable in order to extract surplus from other types from which that type cannot be distinguished. We sketch the argument here, which then serves as the template for the general constructive proof we give below. 

\begin{figure}
    \centering
    \begin{minipage}{0.45\textwidth}
        \centering
\begin{tikzpicture}[scale = 0.6]

\coordinate [label=right: $\pi(0)$] (a) at (4,1);

\coordinate [label=right: $\pi(1)$] (b) at (4,-1);

\coordinate [label = above right: $\pi(t)$] (c) at (3,2.95);

\coordinate [label = left: $z(t)$] (c) at (1.25,3.25);

\fill (4,1) circle (1pt);

\fill (4,-1) circle (1pt);

\draw  [fill = gray!20] (4,1) .. controls (4,4) and (0,4) .. (0,0);

\draw [fill=gray!20] (0,0) .. controls (0,-4) and (4,-4) .. (4,-1);

\draw [gray!20, fill=gray!20] (0,0) -- (4,1) -- (4,-1) -- (0,0) ;

\draw (4,1) -- (4,-1);

\draw (1,3.95) -- (5, 1.95) ;

\draw [->] (1.5, 3.75) -- (1.25, 3.25);

\fill (4,1) circle (2pt);

\fill (4,-1) circle (2pt);

\fill (3,2.95) circle (2pt);

\end{tikzpicture}

\caption{Construction of $c(t)$ when $\pi(t)$ is an exposed point of $C$. }
    \end{minipage}\hfill
    \begin{minipage}{0.45\textwidth}
        \centering
\begin{tikzpicture}[scale = 0.6]

\coordinate [label=right: $\pi(0)$] (a) at (4,1);

\coordinate [label=right: $\pi(1)$] (b) at (4,-1);

\fill (4,1) circle (1pt);

\fill (4,-1) circle (1pt);

\draw  [fill = gray!20] (4,1) .. controls (4,4) and (0,4) .. (0,0);

\draw [fill=gray!20] (0,0) .. controls (0,-4) and (4,-4) .. (4,-1);

\draw [gray!20, fill=gray!20] (0,0) -- (4,1) -- (4,-1) -- (0,0) ;

\draw (4,1) -- (4,-1);

\draw (2.5,4) -- (4, 1) -- (5.5,-2);

\fill (4,1) circle (2pt);

\fill (4,-1) circle (2pt);

\end{tikzpicture}

\caption{Leaving some surplus for types close to $t=0$ might be necessary to extract surplus from type $t=1$. }
    \end{minipage}
\end{figure}

\noindent {\bf Example 1 (continued): } First, consider the case of $t\not\in \{0,1\}$, so $t$ is detectable in $T$ and $\pi(t)$ is an exposed point of $C = \mbox{co} \{ \pi(s): s\in T\}$. Then there exists $z(t) \in {\bf R}^S$ such that 
\[
\pi(t) \cdot z(t) = 0
\]
and
\[
\pi(s) \cdot z(t) >0 \ \ \ \forall s\not= t
\]
See Figure 4. Following the idea of the proof for the finite types case, consider a contract of the form $c(t) = v(t) + \alpha(t) z(t)$. For type $t$, this contract has expected cost $v(t)$:
\[
\pi(t) \cdot c(t) = v(t) + \alpha(t) ( \pi(t)\cdot z(t) ) = v(t)
\]
while for types $s\not= t$, the expected cost is
\[
\pi(s) \cdot c(t) = v(t) + \alpha(t) ( \pi(s) \cdot z(t) ) 
\]
As in the finite case, to make this contract unattractive to types $s\not= t$, we would like to set $\alpha(t) >0$ sufficiently large so that
\[
\alpha (t) > \sup_{s\not= t} \frac{v(s) - v(t)}{\pi(s) \cdot z(t)}
\]
The term on the right need not be finite with infinitely many types, however (indeed, both the denominator and the numerator go to 0 as $s\to t$), so this scaling term need not be defined. For types sufficiently close to $t$, however, the difference in values $v(s) - v(t)$ is small, so the surplus such types could gain by choosing the contract $c(t)$ is small, as 
\[
v(s) - \pi(s) \cdot c(t) = v(s) - v(t) - \alpha(t) ( \pi(s)\cdot z(t) ) < v(s) -v(t)
\]
Given $\varepsilon >0$, we can find $\delta >0$ such that $\Vert s-t\Vert < \delta \Rightarrow v(s) - v(t) < \varepsilon$. Since the set of types in $T$ with $\Vert s-t\Vert \geq \delta$ is compact, we can now choose $\alpha(t) >0$ sufficiently large so that
\[
\alpha (t) > \max_{\Vert s-t\Vert \geq \delta } \frac{v(s) - v(t)}{\pi(s) \cdot z(t)}
\]
For the resulting contract $c(t) = v(t) + \alpha(t) z(t) $, 
\[
v(t) - \pi(t) \cdot c(t) = v(t) - v(t) =0
\] 
and
\[
v(s) - \pi(s) \cdot c(t) \leq \varepsilon \ \ \ \forall s\not =t
\]
Now consider $t\in \{ 0,1\}$, and without loss of generality take $t=0$. First note that, as depicted in Figure 3, there exists $z(0) \in {\bf R}^S$ such that
\[
\pi(0) \cdot z(0) = \pi(1) \cdot z(0) = 0
\] 
and
\[
\pi(t) \cdot z(0) >0 \ \ \ \forall t\not\in\{ 0,1\}
\]
Then consider $T_1 = \{ 0,1\}$ and 
\[
C_1 = \{ \alpha \pi(0) + (1-\alpha) \pi(1) : \alpha \in [0,1]\} = \mbox{co} \{ \pi(0), \pi(1)\}
\]
In $C_1$, $\pi(0)$ is an exposed point, thus there exists $z_1(0) \in {\bf R}^S$ such that 
\[
\pi(0) \cdot z_1(0) = 0 \  \mbox{ and } \ \pi \cdot z_1(0) > 0 \ \ \ \forall \pi\in C_1\setminus \{ \pi(0)\}
\]
In particular, $\pi(1) \cdot z_1(0) >0$. See Figures 6 and 7. Mimicking the construction in the previous case for a detectable type, choose $\alpha_1(0) > 0$ sufficiently large so that 
\[
\alpha_1(0) > \frac{v(1) - v(0) } { \pi(1) \cdot z_1(0)}
\]
Then set $c_1(0) = v(0) + \alpha_1(0) z_1(0)$. 

By construction, 
\[
v(0) - \pi(0) \cdot c_1(0) = v(0) -v(0) -\alpha_1(0) ( \pi(0) \cdot z_1(0) ) = 0
\]
and
\[
v(1) -\pi(1) \cdot c_1(0) = v(1) -v(0) -\alpha_1(0) (\pi(1) \cdot z_1(0) ) <0
\]
The contract $c_1(0)$ might leave surplus to some types $t\not\in \{ 0,1\}$, however. To take care of this, we can use the additional stochastic payment $z(0)$, which type 0 and type 1 believe has expected value of zero. Since all other types believe $z(0)$ has positive expected value, we can appropriately scale $z(0)$ to ensure that no more than $\varepsilon$ expected surplus is left for all such types, as follows.

\begin{figure}
\centering
\begin{tikzpicture}[scale = 0.6]

\coordinate [label=right: $\pi(0)$] (a) at (4,1);

\coordinate [label=right: $\pi(1)$] (b) at (4,-1);

\draw (4,1) -- (4,-1);

\fill (4,1) circle (2pt);

\fill (4,-1) circle (2pt);

\end{tikzpicture}

\caption{Graph of $C_1 = \{ \alpha \pi(0) + (1-\alpha) \pi(1) : \alpha \in [0,1]\}$. }
\end{figure}

\begin{figure}
\centering
\begin{tikzpicture}[scale = 0.6]

\coordinate [label=right: $\pi(0)$] (a) at (4,1);

\coordinate [label=right: $\pi(1)$] (b) at (4,-1);

\coordinate [label = left: $z_1(0)$] (c) at (2.75,2.25);

\draw [->] (3.25, 2.50) -- (2.75, 2.25);

\draw (4,1) -- (4,-1);

\draw (3,3) -- (4, 1) -- (4.75,-0.5);

\fill (4,1) circle (2pt);

\fill (4,-1) circle (2pt);

\end{tikzpicture}

\caption{Construction of $c_1(0)$. }
\end{figure}

Note that the surplus $c_1(0)$ leaves to types arbitrarily close to 0 must be arbitrarily small, since $c_1(0)$ leaves zero surplus for type 0. Similarly, the surplus $c_1(0)$ leaves to types arbitrarily close to 1 must be negative, since $c_1(0)$ leaves negative surplus for type 1. Then choose $\delta >0$ such that for $\Vert s-0\Vert < \delta$ or $\Vert s-1\Vert <\delta$,
\[
v(s) - \pi(s) \cdot c_1(0) = v(s) - v(0) -\alpha_1(0) (\pi(s) \cdot z_1(0) )  <\varepsilon
\]
This is possible, by the continuity of $v$ and $\pi$, and the fact that  $v(0) -\pi(0) \cdot c_1(0) = 0$ and $v(1)- \pi(1)\cdot c_1(0) <0$. Then set $\alpha(0) >0$ such that
\[
\alpha(0) > \max_{\substack{\Vert s-0\Vert \geq \delta \\ \Vert s-1\Vert \geq \delta}} \frac{ v(s) - \pi(s) \cdot c_1(0)}{\pi(s) \cdot z(0)}
\]
Then set 
\[
c(0) = c_1(0) + \alpha(0) z(0) = v(0) + \alpha_1(0) z_1(0) + \alpha(0) z(0)
\]
By construction, $\pi(0) \cdot c(0) = \pi(0) \cdot c_1(0)$  and $\pi(1) \cdot c(0) = \pi(1) \cdot c_1(0)$. Thus 
$v(0) - \pi(0) \cdot c(0) = 0$ and $v(1) -\pi(1) \cdot c(0) <0$. 

For $t\not\in \{ 0,1\}$, again by construction, 
\[
v(t) -\pi(t) \cdot c(0) = v(t) -\pi(t) \cdot c_1(0) -\alpha(0) (\pi(t) \cdot z(0) ) < \varepsilon
\]
The collection $\{ c(t) : t\in T\}$ thus constructed achieves extraction of all but at most $\varepsilon$ surplus.  \hfill$\diamondsuit$

\bigskip

Next we will show that an analogous construction works in general whenever probabilistic independence is satisfied to yield a collection of contracts that achieves virtual extraction, provided the state space is finite.  
Three key aspects of the construction in the example lead to the general result. First, for any detectable type there is always a contract that leaves that type with zero expected surplus while leaving at most $\varepsilon$ expected surplus for all other types. Second, although some types might not be detectable, such as types 0 and 1 in the example, probabilistic independence implies that every type is eventually detectable within some (non-singleton) subset of types, in a sense we make precise below. This step of the construction requires a finite state space, while the others hold in general. Finally, this weaker property is sufficient to construct a suitable analogous contract for a type that is not detectable. We formalize this next, using additional results about convex sets and their extreme points that help illuminate the connection between the extraction problems with finitely and infinitely many types. 

We start with a weaker notion of exposed point for a convex set, motivated by the above construction of contracts for the case when types are not detectable. 

\begin{definition}
Let $C \subseteq {\bf R}^k$ be a convex set. A point $x\in C$ is \emph{eventually exposed} if there exists a sequence $\{ F_0, F_1, \ldots ,F_n\}$ of subsets of $C$ such that
\begin{itemize}
  \item[(i)]  $F_0 = C$, $F_n = \{ x\}$
  \item[(ii)] if $n>1$, then $\mbox{ dim } F_{i} \geq 1$ for $i=1,\ldots ,n-1$ 
  \item[(iii)] for each $i$, $F_{i+1} \subseteq F_i$ and $F_{i+1}$ is an exposed set in $F_i$
\end{itemize}
\end{definition}

\noindent {\bf Note:} If $C\subseteq {\bf R}^k$ is a convex set and  $x\in C$ is an exposed point of $C$, then $x$ is also eventually exposed, using the trivial sequence $F_0 = C$  and $F_1 = \{x\}$. 

\noindent {\bf Note:} If $C\subseteq {\bf R}^k$ is compact and convex and $x\in C$ is eventually exposed, then each set $F_i$ is compact and convex. In addition,  $x$ is an exposed point in the set $F_{n-1}$.  

Next, while a convex set can have extreme points that are not exposed, every extreme point in a compact, convex subset of ${\bf R}^k$ is eventually exposed. We include a proof of this result in the appendix for completeness.\footnote{See also e.g. Soltan (2015, Theorem 12.7). We thank Yeon-Koo Che for this reference.}

\begin{theorem}
Let $C$ be a compact, convex subset of ${\bf R}^k$. If $x\in C$ is an extreme point then it must be eventually exposed. 
\end{theorem}

To illustrate using Example 1, note that while $\pi(0)$ and $\pi(1)$ are not exposed points of $C =\mbox{co} \{ \pi(s) : s\in T \}$, each is eventually exposed. To see this for $\pi(0)$, for example, let $F_0 = C$, $F_1 = \{ \alpha \pi(0) + (1-\alpha) \pi(1): \alpha \in [0,1]\}$, and $F_2 = \{\pi(0)\}$. Then $F_1\subseteq C$ is an exposed set in $C$ with $\mbox{ dim } F_1 =1$, and $\pi(0) \in F_1$ is an exposed point of $F_1$.  

Returning to the extraction problem, these definitions and results suggest analogous conditions connecting types to beliefs, and the geometry of the set of beliefs, as with the connection between exposed points and detectable types. First, we extend the notion of a detectable type to a set of types in the natural way. 

\begin{definition}
A set of types $T^*\subseteq T$ is \emph{detectable} in $T$  if there exists $z\in C(S)$ such that 
\begin{eqnarray*}
\pi(t) \cdot z &=& 0 \ \ \forall t\in T^*\\
\pi(t) \cdot z &>& 0 \ \ \forall t\in T\setminus T^*
\end{eqnarray*} 
\end{definition}

Next, in analogy with the connection between exposed points and eventually exposed points, there is a natural notion of eventually detectable types. 

\begin{definition}
A type $t^*\in T$ is \emph{eventually detectable} if there is a nested sequence of compact sets $\{ T_0, T_1, \ldots, T_n\}$ with $T_0 = T$, $T_n= \{ t^*\}$, and $T_{i+1} \subseteq T_i$ for each $i$, and a corresponding sequence $\{ z_1, \ldots , z_n\} \subseteq C(S)$ such that 
for each $i=1,\ldots ,n$, 
  \begin{eqnarray*}
\pi(t) \cdot z_i &=& 0 \ \ \forall t\in T_i\\
\pi(t) \cdot z_i &>& 0 \ \ \forall t\in T_{i-1}\setminus T_i
\end{eqnarray*} 
\end{definition}

\noindent {\bf Note:} A type $t^*\in T$ is eventually detectable if there is a nested sequence of compact sets $\{ T_0, T_1, \ldots ,T_n\}$, with $T_0=T$, $T_n = \{ t^*\}$, and $T_{i+1} \subseteq T_i$ for each $i$ such that $T_{i+1}$ is detectable in $T_i$ for each $i$. 

\noindent {\bf Note: } If $t^*\in T$ is detectable, then it is eventually detectable, using the trivial sequence $T_0 = T$, $T_1 = \{ t^*\}$. 

Illustrating using Example 1 again, note that the type $t=0$ is not detectable, but is eventually detectable, using the sequence $T_0 = T$, $T_1 = \{ 0,1\}$, $T_2 = \{0\}$. 

The following lemma provides two key results that connect these ideas to the extraction problem. First, when a type $t^*$ is detectable in some compact subset of types $T^*$, then for any $\varepsilon >0$ a contract can be constructed that leaves type $t^*$ zero expected surplus while leaving the other types in $T^*$ expected surplus at most $\varepsilon$. Second, building on this observation, if a type is eventually detectable, such a contract can be constructed to leave expected surplus at most $\varepsilon$ for all types in $T$. 

\begin{lemma}
Let $v:T\to {\bf R}$ be given. 
\begin{itemize}
  \item[(i)] Let $T^*\subseteq T$ be compact and $t^* \in T^*$. If $t^*$ is detectable in $T^*$, then for every $\varepsilon >0$ there exists $c(t^*) \in C(S)$ such that $v(t^*) - \pi(t^*)\cdot c(t^*) = 0$ and $v(t) - \pi(t) \cdot c(t^*) <\varepsilon$ for all $t\in T^*\setminus \{ t^*\}$. 
  \item[(ii)]  If $t^* \in T$ is eventually detectable, then for every $\varepsilon >0$ there exists $c(t^*)\in C(S)$ such that $v(t^*) - \pi(t^*)\cdot c(t^*) = 0$ and $v(t) - \pi(t) \cdot c(t^*) <\varepsilon$ for all $t\in T\setminus \{ t^*\}$. 
\end{itemize} 
 \end{lemma}
\begin{proof}
For (i), since $t^*$ is detectable in $T^*$, choose $z(t^*)\in C(S)$ such that 
\begin{eqnarray*}
\pi(t^*) \cdot z(t^*) &=& 0\\
\pi(t) \cdot z(t^*) &>& 0 \ \forall t\in T^*\setminus \{t^*\}
\end{eqnarray*}
Fix $\varepsilon >0$. Choose $\delta >0$ such that $\Vert t-t^* \Vert <\delta \Rightarrow \Vert v(t) - v(t^*) \Vert <\varepsilon$. Set $\alpha(t^*) >0$ such that 
\[
\alpha(t^*) > \max_{\{ t\in T^*: \Vert t-t^*\Vert \geq \delta \}} \frac{v(t)-v(t^*)}{\pi(t)\cdot z(t^*)}
\]
Note that $\alpha(t^*)$ is well-defined (i.e., finite) since $\{ t\in T^*: \Vert t-t^*\Vert \geq \delta\}$ is compact,  $v$ and $\pi$ are continuous, and $\pi(t)\cdot z(t^*) >0$ for all $t\in T^*\setminus\{ t^*\}$. 

Then define a contract $c(t^*) \in C(S)$ as follows (here and in what follows, we use the constant $r\in {\bf R}$ interchangeably with the constant function  $r{\bf 1}(S)$):
\[
c(t^*) = v(t^*) + \alpha(t^*)z(t^*)
\]

Now for type $t^*$:
\[
v(t^*) - \pi(t^*)\cdot c(t^*) = v(t^*) - v(t^*) -\alpha(t^*)  ( \pi(t^*)\cdot z(t^*) ) = 0
\]

For types $t\in T^*\setminus \{t^*\}$, first suppose $\Vert t-t^*\Vert \geq \delta$. Then
\[
v(t) - \pi(t)\cdot c(t^*) = v(t) - v(t^*) -\alpha(t^*) ( \pi(t)\cdot z(t^*) ) < 0 
\]
by definition of $\alpha(t^*)$. If instead $\Vert t-t^* \Vert <\delta$, then 
\begin{eqnarray*}
v(t) - \pi(t)\cdot c(t^*) &=& v(t) - v(t^*) -\alpha(t^*) ( \pi(t)\cdot z(t^*) )\\
&<& v(t) - v(t^*) \ \mbox{ since } \pi(t)\cdot z(t^*) >0 \mbox{ and } \alpha(t^*) >0 \\
&<& \varepsilon  \mbox{ since } \Vert t-t^* \Vert <\delta
\end{eqnarray*}
Thus for all types $t\in T^*\setminus \{t^*\}$, $v(t) -\pi(t)\cdot c(t^*) <\varepsilon$. This establishes the claim in (i). 

For (ii), since $t^*$ is eventually detectable in $T$, there exists a nested sequence of compact sets $\{ T_0, T_1, \ldots, T_n\}$ with $T_0 = T$, $T_n= \{ t^*\}$, and $T_{i+1} \subseteq T_i$ for each $i$, and a corresponding sequence $\{ z_1, \ldots , z_n\} \subseteq C(S)$ such that 
for each $i=1,\ldots ,n$, 
  \begin{eqnarray*}
\pi(t) \cdot z_i &=& 0 \ \forall t\in T_i\\
\pi(t) \cdot z_i &>& 0 \ \forall t\in T_{i-1}\setminus T_i
\end{eqnarray*} 

Fix $\varepsilon >0$. Now consider $T_{n-1}$ and $t^*$. Since $t^*$ is detectable in $T_{n-1}$, by (i) there is a contract $c_{n-1}(t^*)$ such that 
\[
v(t^*) - \pi(t^*)\cdot c_{n-1}(t^*) = 0
\]
and
\[
v(t) - \pi(t) \cdot c_{n-1}(t^*) < \frac{1}{n} \varepsilon \ \ \forall t\in T_{n-1}
\]
Now we claim, by induction, that for each $k=0, \ldots ,n-1$, there is a contract $c_k(t^*)$ such that 
\[
v(t^*) - \pi(t^*) \cdot c_k(t^*) =0
\]
and
\[
v(t) - \pi(t) \cdot c_k(t^*) < \frac{n-k}{n} \varepsilon \ \ \forall t\in T_k
\]
To see this, fix $i\geq 0$ and suppose there exists $c_i(t^*)$ such that 
\[
v(t^*) - \pi(t^*) \cdot c_i(t^*) =0
\]
and
\[
v(t) - \pi(t) \cdot c_i(t^*) < \frac{n-i}{n} \varepsilon \ \ \forall t\in T_i
\]
Then we claim that there exists $c_{i-1}(t^*)$ such that 
\[
v(t^*) - \pi(t^*) \cdot c_{i-1}(t^*) = 0
\]
and
\[
v(t) - \pi(t) \cdot c_{i-1}(t^*) < \frac{n-(i-1)}{n} \varepsilon \ \ \forall t \in T_{i-1}
\]
To show this, first note that for each $t\in T_i$, there exists $\delta(t) > 0$ such that $\Vert s-t\Vert < \delta(t) \Rightarrow $
\[
\Vert v(t) - \pi(t) \cdot c_i(t^*) - [v(s) - \pi(s) \cdot c_i (t^*)]\Vert  < \frac{1}{n}\varepsilon
\]
For each $t \in T_i$, let $B_{\delta(t)}(t) = \{ s\in T: \Vert s - t\Vert < \delta(t)\}$. 
The collection $\{ B_{\delta(t)}(t): t\in T_i \}$ is an open cover of the compact set $T_i$, so there exists $\{ t^1, \ldots ,t^m\} \subseteq T_i$ such that $T_i \subseteq \cup_j B_{\delta(t^j)}(t^j)$. Moreover, $\cup_j B_{\delta(t^j)}(t^j)$ is open, so $T_{i-1}\setminus \cup_j B_{\delta(t^j)}(t^j)$ is compact. Choose $R_i>0$ such that
\[
R_i > \max_{t\in T_{i-1}\setminus \cup_j B_{\delta(t^j)}(t^j)} v(t) -\pi(t) \cdot c_i(t^*)
\]
and choose $\alpha_i(t^*)>0$ such that 
\[
\alpha_i(t^*) > \max_{t\in T_{i-1}\setminus \cup_j B_{\delta(t^j)}(t^j)} \frac{R_i}{\pi(t) \cdot z_i(t^*)}
\]
Then set 
\[
c_{i-1}(t^*) = c_i(t^*) + \alpha_i(t^*) \cdot z_i(t^*)
\]
For each $t\in T_i$:
\begin{eqnarray*}
v(t) -\pi(t) \cdot c_{i-1}(t^*) &=& v(t) - \pi(t) \cdot c_i(t^*)\\
&<& \frac{n-i}{n}\varepsilon\\
&<& \frac{n-(i-1)}{n} \varepsilon
\end{eqnarray*}
and
\[
v(t^*) - \pi(t^*) \cdot c_{i-1}(t^*) = v(t^*) -\pi(t^*) \cdot c_i(t^*) = 0
\]
For $t\in T_{i-1}\setminus T_i$:
\begin{itemize}
  \item if $t\in \cup_j B_{\delta(t^j)}(t^j)$:
  \begin{eqnarray*}
v(t) - \pi(t)\cdot c_{i-1}(t^*) &=& v(t) -\pi(t) \cdot c_i(t^*) - \alpha_i (t^*) ( \pi(t) \cdot z_i(t^*) ) \\
&<& v(t) - \pi(t) \cdot c_i(t^*) \ \ \mbox{ since } \pi(t)\cdot z_i(t^*) >0\\
&<& \frac{n-i}{n}\varepsilon + \frac{1}{n} \varepsilon  \ \ \mbox{ by construction }\\
& = & \frac {n-(i-1)}{n} \varepsilon 
\end{eqnarray*}

\item if $t\in T_{i-1} \setminus \cup_j B_{\delta(t^j)}(t^j)$: 
\begin{eqnarray*}
v(t) -\pi(t) \cdot c_{i-1}(t^*) &=& v(t) -\pi(t) \cdot c_i(t^*) - \alpha_i(t^*) ( \pi(t) \cdot z_i(t^*) )\\
&<& 0 \ \ \ \mbox{ by choice of } \alpha_i(t^*)
\end{eqnarray*}
\end{itemize}
Thus for all $t\in T_{i-1}\setminus \{ t^*\}$, 
\[
v(t) -\pi(t) \cdot c_{i-1}(t^*) < \frac{n-(i-1)}{n} \varepsilon
\]
and
\[
v(t^*) -\pi(t^*) \cdot c_{i-1}(t^*) = 0
\]
Thus by induction, for each $k$ there exists a contract $c_k(t^*) $ such that 
for all $t\in T_k \setminus \{ t^*\}$, 
\[
v(t) -\pi(t) \cdot c_k(t^*) < \frac{n-k}{n} \varepsilon
\]
and
\[
v(t^*) -\pi(t^*) \cdot c_k(t^*) = 0
\]
In particular, consider $k=0$: there exists $c(t^*) = c_0(t^*)$ such that 
 for all $t\in T_0\setminus \{ t^*\}= T\setminus \{ t^*\}$, 
\[
v(t) -\pi(t) \cdot c(t^*) < \frac{n}{n} \varepsilon = \varepsilon
\]
and
\[
v(t^*) -\pi(t^*) \cdot c(t^*) = 0
\]
This establishes the claim.
\end{proof}

From Lemma 1, it follows immediately that virtual extraction holds whenever all types are eventually detectable. We show next that when $S$ is finite and types satisfy probabilistic independence, then while types need not be detectable, every type is eventually detectable, which in turn guarantees that virtual extraction holds. 

\begin{theorem}

\begin{itemize}
  \item[(i)] If every type is eventually detectable, then virtual extraction holds.
  \item[(ii)] Let $S$ be finite. If types satisfy probabilistic independence, then every type is eventually detectable. 
\end{itemize} 
\end{theorem}

\begin{proof}
Part (i) follows from part (ii) of Lemma 1. For (ii), let $t^*\in T$. By probabilistic independence, $\pi(t^*)$ is an extreme point of $C = \mbox{co} \{ \pi(s) : s\in T \}$. Since $\{ \pi(s) : s\in T\} \subseteq {\bf R}^S$ is compact, its convex hull is closed, thus $C = \overline{\mbox{co}} \{ \pi(s) : s\in T \}$ and $C$ is compact. If $t^*$ is not detectable, then $\pi(t^*)$ is an extreme point of $C$ that is not exposed. Hence by Theorem 4, $\pi(t^*)$ is eventually exposed in $C$. Thus there is a sequence $\{ F_0,F_1, \ldots ,F_n\}$ with $F_0 =C$, $F_n = \{ \pi(t^*)\}$, and a corresponding sequence  $\{ z_1, \ldots ,z_n\}$ such that for each $i$, $\pi(t^*) \in F_i$ and
\begin{eqnarray*}
\pi\cdot z_i &=& 0 \ \ \forall \pi \in F_i\\
\pi\cdot z_i &>& 0 \ \ \forall \pi \in F_{i-1}\setminus F_i
\end{eqnarray*}
Set $T_0 = T$, and 
\[
T_1 = \{ t\in T: \pi(t) \in F_1\}
\]
Note that since $T$ and $F_1$ are compact and the map $t\mapsto \pi(t)$ is continuous, $T_1$ is compact. Then for each $i\geq 2$, set 
\[
T_i = \{ t\in T_{i-1}: \pi(t) \in F_i\}
\]
By induction, $T_i$ is compact for each $i$, and by construction $T_i \subseteq T_{i-1}$ for each $i$. Also by construction, 
\[
T_n = \{ t\in T_{n-1}: \pi(t) \in F_n \} = \{ \pi(t^*) \} = \{ t^*\}
\]
Then let $m$ be the minimum index $i$ for which $T_i = \{ t^*\}$. Consider the sequence $\{ T_0, T_1, \ldots ,T_m\}$ and the corresponding sequence $\{ z_1, \ldots, z_m\}$. Then for each $i=1,\ldots ,m$,
\begin{eqnarray*}
\pi(t)\cdot z_i &=& 0 \ \ \forall t\in T_i\\
\pi(t) \cdot z_i &>& 0 \ \ \forall t\in T_{i-1}\setminus T_i
\end{eqnarray*}
Thus $t^*$ is eventually detectable.
\end{proof}

\bigskip

As Example 1 illustrated, full extraction requires that all types are detectable. With infinitely many types however, this is not sufficient for full extraction. We close this section by noting a stronger condition which is sufficient for full extraction in general, and then see probabilistic independence implies this stronger condition with finitely many types.\footnote{ See McAfee and Reny (1992) for a discussion of several other sufficient conditions for full extraction when $T$ is infinite.}  

\begin{definition}
A type $t^* \in T$ is \emph{strongly detectable} if there exists $z\in C(S)$ such that 
\[
\pi(t^*) \cdot z = 0
\]
and
\[
\inf_{t\not= t^*} \pi(t) \cdot z >0
\]
\end{definition}

\bigskip

Note that if all types are strongly detectable, then full extraction is possible for any $t\mapsto v(t)$, and any type set $T$. This follows from the basic argument for the finite case, as sketched and adapted above. Also note that when $T$ is finite and types satisfy probabilistic independence, which is equivalent to convex independence because $T$ is finite, then all types are strongly detectable. We record these observations below.

\begin{theorem}
\begin{itemize}
  \item[(i)]  If every type is strongly detectable, then full extraction holds for any $v:T\to {\bf R}$. 
  \item[(ii)] If $T$ is finite and types satisfy probabilistic independence, then every type is strongly detectable. 
\end{itemize}
\end{theorem}
\begin{proof}
Part (i) is a straightforward adaptation of arguments above; we omit the details. For (ii), fix $t\in T$. Types satisfy probabilistic independence, so $\pi(t) \not\in \overline{\mbox{co}} \{ \pi(s) : s\in T, s\not= t\}$. Since $T$ is finite, $\overline{\mbox{co}} \{ \pi(s) : s\in T, s\not= t\}$ is compact. Thus there exists $z(t) \in C(S)$ such that 
$\pi(t) \cdot z(t) = 0$ and $\pi(s) \cdot z(t) >0$ for all $s\in T, \ s\not= t$. Since $T$ is finite, this implies $\inf_{s\not= t} \pi(s) \cdot z(t) >0$. Thus $t$ is strongly detectable by definition. Repeating for each $t\in T$ yields the result.  
\end{proof}

\section{Duality and General State Space}

In this section we consider the general case in which the state space $S$ is an arbitrary compact metric space. For example, this allows the model to accommodate the case in which $S=T^n$ for some $n$, or more generally when the cardinality of $S$ is greater than or equal to the cardinality of $T$.  

We start by again considering the problem of full extraction, and cast the problem in slightly stronger terms (we will see that while this gives a stronger condition, probabilistic independence guarantees that virtual extraction holds under this stronger condition). Rather than looking for a collection of contracts $\{ c(t) \in C(S) : t\in T\}$, we add the requirement that the contracts also be jointly continuous in types and states, and thus consider the existence of a schedule of contracts
$c\in C(T\times S)$ such that for each $t\in T$:
\begin{eqnarray*}
v(t) - \pi(t) \cdot c(t) &=& 0 \\
v(t) - \pi(t)\cdot c(s) &\leq & 0 \ \ \ \forall s\not= t
\end{eqnarray*}
Considering contracts $c\in C(T\times S)$ is  useful because this helps ensure continuity in several constructions below, and because $C(T\times S)$ is a separable Banach space, which is used in a number of steps below. 

First, we note that full extraction is equivalent to the seemingly weaker condition {\it weak full extraction}: for each $t\in T$,
\begin{eqnarray*}
v(t) - \pi(t) \cdot c(t) &\geq& 0 \\
v(t) - \pi(t)\cdot c(s) &\leq & 0 \ \ \ \forall s\not= t
\end{eqnarray*}
We establish this equivalence in the lemma below. 

\begin{lemma}
For each $v:T\to {\bf R}$,  $c\in C(T\times S)$ satisfies full extraction if and only if $c$ satisfies weak full extraction.
\end{lemma}
\begin{proof}
Full extraction clearly implies weak full extraction. To see that these are equivalent, fix $v:T\to {\bf R}$ and  $t\in T$. Then choose a sequence $s_n\to t$ with $s_n \not= t$ for each $n$; this is possible by the connectedness of $T$.  Suppose $c\in C(T\times S)$ satisfies weak full extraction. Then for each $n$, since $s_n \not= t$,  
\[
v(s_n) -\pi(s_n) \cdot c(t) \leq 0
\]
and
\[ 
v(t) - \pi(t) \cdot c(t) \geq 0
\]
Since $v(s_n)\to v(t)$ and $\pi(s_n)\to \pi(t)$, 
\[
v(t) - \pi(t) \cdot c(t) \leq 0
\]
Thus 
\[
v(t) -\pi(t) \cdot c(t) = 0
\]
Since $t\in T$ was arbitrary, the equivalence follows.
\end{proof}

From Lemma 2, it is enough to consider the relaxed problem of weak full extraction.\footnote{Lemma 2 uses the assumption that $T$ is connected. If $T$ is an arbitrary compact metric space, it is not difficult to show instead that if $c\in C(T\times S)$ satisfies weak full extraction, then there exists $c'\in C(T\times S)$ satisfying full extraction; see Lemma 9 in the Appendix. } Now write 
\[
c(t) = v(t) + z(t) 
\]
where $z\in C(T\times S)$, and we use $v(t) \in {\bf R}$ interchangeably with $v(t) {\bf 1}(S)$, where ${\bf 1}(S)$ denotes the identity on $S$. Note that any $c\in C(T\times S)$ can be written this way for appropriate choice of $z$. Then $c$ satisfies full extraction if and only if for each $t\in T$, 
\begin{eqnarray*}
\pi(t) \cdot z(t) &\leq & 0 \\
\pi(t) \cdot z(s) &\geq & v(t) - v(s) \ \ \  \forall s\not= t
\end{eqnarray*}
This follows from observing that for each $t$,
\[
v(t) - \pi(t) \cdot c(t) = v(t) - v(t) - \pi(t) \cdot z(t) = -\pi(t) \cdot z(t) 
\]
and
\[
v(t) - \pi(t) \cdot c(s) = v(t) - v(s) - \pi(t) \cdot z(s) 
\]
Thus we will consider the existence of $z\in C(T\times S)$ such that for each $t\in T$, 
\begin{eqnarray*}
\pi(t) \cdot z(t) &\leq & 0 \\
\pi(t) \cdot z(s) &\geq & v(t) - v(s) \ \ \  \forall s\not= t
\end{eqnarray*}

To that end, let $f:T\times C(T\times S) \to {\bf R}$ be given by 
\[
f(t,z) =  \pi(t) \cdot z(t)
\]
and for each $t\in T$, let $f_t:C(T\times S) \to {\bf R}$ be given by $f_t(z) = f(t,z)$. Similarly, let $g:T\times T\times C(T\times S)\to {\bf R}$ be given by 
\[
g_{(s,t)}(z) := g(s,t,z) =  \pi(s) \cdot z(t)
\]

Then to show that full extraction is possible, it suffices to show that there exists $z\in C(T\times S)$ such that 
\begin{eqnarray*}
f(t,z) & \leq & 0 \ \ \forall t\in T\\
g(t,s,z) &\geq & v(t) - v(s) \ \ \forall s\not= t, \ \forall t\in T
\end{eqnarray*}

\noindent {\bf Note: } For each $t\in T$, $f_t$ is convex, and for every $t,s \in T$, $g_{(t,s)}$ is concave. (Both are in fact linear.) In addition, $f(t,0) = g(t,s,0) = 0$ for every $t,s\in T$. 

\bigskip

Now consider the problem 

\begin{equation}
\begin{aligned}
p^* :=& \underset{c\in {\bf R}, \; z\in C(T\times S)}{\text{inf}}
& & c \\
& \text{subject to}
& & f(t,z) \leq c \ \ \ \forall t\in T \\
&&& v(t) - v(s) -g(t,s,z) \leq c \ \ \ \forall t,s \in T
\end{aligned}
\tag{vse}
\end{equation}

Note that if the optimal value $p^*$ of this problem is less than or equal to zero, then at least virtual surplus extraction is possible. We establish this in the next lemma. 

\begin{lemma}
If $p^*\leq 0$, then virtual surplus extraction holds. If $p^*<0$, or if $p^*=0$ and is attained in (vse), then full extraction holds. 
\end{lemma}
\begin{proof}
 To see this, first suppose $p^*<0$. Then there must exist $z\in C(T\times S)$ such that 
\[
f(t, z) \leq 0 \ \ \ \forall t\in T
\]
and
\[
v(t) - v(s) - g(t,s,z) \leq 0 \ \ \ \forall t,s\in T
\]
In this case, full extraction holds, and thus a fortiori, virtual extraction holds as well. Similarly, if $p^* = 0$, then either $p^*$ is attained, in which case again there must exist such a $z\in C(T\times S)$ as above so that full extraction holds, or if $p^*$ is not attained, then for each $\varepsilon >0$ there exists $z_\varepsilon \in C(T\times S)$ such that 
\[
f(t, z_\varepsilon) \leq \varepsilon \ \ \ \forall t\in T
\]
and
\[
v(t) - v(s) - g(t,s,z_\varepsilon) \leq \varepsilon \ \ \ \forall t,s\in T
\]
Now note that setting $s=t$, this implies
\[
\pi(t) \cdot z_\varepsilon (t)  \geq -\varepsilon \ \ \ \forall t\in T
\]
Thus
\[
-\varepsilon \leq \pi(t) \cdot z_\varepsilon (t) \leq \varepsilon \ \  \forall t\in T
\]
Now set 
\[
z:= z_\varepsilon - \varepsilon
\]
and for each $t\in T$, set the contract $c(t)$ to be 
\[
c(t) = v(t) + z(t)
\]
Then $c\in C(T\times S)$, and for each $t\in T$, 
\begin{eqnarray*}
v(t) - \pi(t) \cdot c(t) &=& v(t) - v(t) -\pi(t) \cdot z(t) \\
&=& -\pi(t) \cdot z(t) \\
&=& \varepsilon - \pi(t) \cdot z_\varepsilon (t)
\end{eqnarray*}
and by the preceding argument,
\[
0 \leq \varepsilon - \pi(t) \cdot z_\varepsilon (t)    \leq 2\varepsilon
\]
Thus for each $t\in T$, 
\[
0\leq v(t) -\pi(t) \cdot c(t) \leq 2\varepsilon 
\]
Then fix $t\in T$, and consider $s\not= t$. 
\begin{eqnarray*}
v(t) - \pi(t) \cdot c(s) &=& v(t) - v(s) -\pi(t) \cdot z(s) \\
&=& v(t) - v(s) -\pi(t) \cdot z_\varepsilon (s) + \varepsilon\\
&\leq & v(t) -v(s) -g(t,s,z_\varepsilon) + \varepsilon\\
&\leq & 2\varepsilon
\end{eqnarray*}
Thus for all $t\in T$, 
\[
v(t) - \pi(t) \cdot c(s) \leq 2\varepsilon \ \ \  \forall s\not= t
\]
So for all $t\in T$,
\[
0\leq \sup_{s\in T} v(t) - \pi(t) \cdot c(s) \leq 2\varepsilon
\]
The result follows. 
\end{proof}

Thus to show that virtual surplus extraction is possible, it suffices to show that $p^*\not> 0$. We establish this below by considering the dual of the optimization problem (vse), and making use of duality to argue that these problems have the same value. The heart of the proof is then to show that this common value cannot be positive under probabilistic independence. 

\begin{theorem}
Let $S$ be a compact metric space. If types satisfy probabilistic independence, then virtual extraction holds. 
\end{theorem} 

\begin{proof}
By Lemma 3, to show that virtual surplus extraction is possible it suffices to show that $p^* \not> 0$. To that end, note that the Lagrange dual function for the problem $(\mbox{vse})$ is 
\[
{\cal L}(\lambda, \nu) = \inf_{{ c\in {\bf R} }\atop {z\in C(T\times S)}} \left\lbrace c + \lambda \cdot (f(t,z) -c) + \nu \cdot (v(t) - v(s) -g(t,s,z) - c) \right\rbrace
\]
where $\lambda \in {\cal M}(T)$ and $\nu\in {\cal M}(T\times T)$. 
Let $d\in C(T\times T)$ be given by 
\[
d(t,s) = v(t) - v(s)
\]
Note that $d(t,t) = 0$ for all $t\in T$. 

Define 
\[
h(\lambda, \nu ) = \inf_{z\in C(T\times S)} \left\{ \lambda \cdot f(z) + \nu\cdot (d-g(z))\right\}
\]
where for $z\in C(T\times S)$, $f(z):T\to {\bf R}$ denotes the function $f(z)(t) = f(t,z)$ for each $t\in T$ and $g(z):T\times T \to {\bf R}$ denotes the function $g(z)(t,s) = g(t,s,z)$ for each $t, s\in T$. 

Using this notation, we can rewrite the Lagrange dual function for $(\mbox{vse})$ as follows:
\[
{\cal L}(\lambda, \nu) = \left\lbrace 
\begin{array}{lr}
h(\lambda, \nu) & \text{ if } \int \lambda (dt) + \iint \nu(ds\ dt) = 1\\
-\infty & \text{ otherwise }
\end{array}
\right.
\]
Thus the dual problem of $(\mbox{vse})$ is 
\begin{equation}
\begin{aligned}
d^* := & \underset{\lambda \in {\cal M}(T), \; \nu\in {\cal M}(T\times T)}{\text{sup}}
& & h(\lambda, \nu) \\
& \text{subject to}
& & (\lambda, \nu) \geq 0 \\
&&& \int \lambda (dt) + \iint \nu(ds\ dt) = 1 
\end{aligned}
\tag{d-vse}
\end{equation}

Then note that Slater's condition holds for the original problem $(\mbox{vse})$. To see this, set $z=0$, so 
\[
f(t,z) = g(t,s,z) = 0 \ \ \ \forall t,s \in T
\]
Then choose 
\[
\bar c > \sup_{t,s\in T} v(t) - v(s) \geq 0
\]
For $(z,c) = (0, \bar c)$, 
\[
\sup_{t\in T} f(t,z) - c = -\bar c <0
\]
and 
\[
\sup_{t,s\in T} v(t) - v(s) - g(t,s,z) -c = \sup_{t,s\in T} v(t) - v(s) -\bar c <0
\]
Thus $p^*=d^*$ and in addition $d^*$ is obtained, where $p^*$ is the optimal value of $(\mbox{vse})$ and $d^*$ is the optimal value of $(\mbox{d-vse})$. 

Now it suffices to show that $p^* = d^* \not> 0$. To show this, suppose by way of contradiction that $p^* = d^* >0$. Since $d^*$ is obtained in $(\mbox{d-vse})$, there exists $(\lambda, \nu) \geq 0$ such that 
\[
d^* = h(\lambda, \nu) >0 \mbox{ and } \int \lambda(dt) + \iint \nu(ds\ dt) =1
\]
Recall that, by definition, 
\[
h(\lambda, \nu) = \inf_{z\in C(T\times S)} \left( \lambda \cdot f(z) + \nu \cdot (d-g(z) ) \right)
\]
and $f(0) = g(0) =0$, which implies
\[
h(\lambda , \nu) \leq \nu \cdot d
\]
Since $h(\lambda, \nu) >0$, this implies $\nu\cdot d >0$. Thus $\nu\not= 0$. Since $\lambda, \nu \geq 0$, this implies $\nu>0$. 

Let $F:C(T\times S) \to {\bf R}$ be given by
\[
F(z) = \lambda \cdot f(z) + \nu\cdot (d-g(z))
\]
Note that $F$ is convex and continuous, and by definition, 
\[
h(\lambda, \nu) = \inf_{z\in C(T\times S)} F(z) >0
\]
In particular, this implies $\inf_{z\in C(T\times S)} F(z) \in {\bf R}$. By Ekeland's Variational Principle (see Lemma 4 in the Appendix), there exists a sequence $\{ z_n\}$ and a sequence $\{ \gamma_n\}$ with $\gamma_n \in \partial F(z_n)$ for each $n$ such that
\[
F(z_n) = \lambda\cdot f(z_n) + \nu\cdot (d-g(z_n) ) \to \inf_{z\in C(T\times S)} F(z)= h(\lambda, \nu) >0
\]
and
\[
\Vert \gamma_n \Vert \to 0
\]
By Lemma 8 (see the Appendix), since $\gamma_n \in \partial F(z_n)$ for each $n$,  $\gamma_n$ is the measure for which 
\[
\gamma_n\cdot y = \int \pi(t) \cdot y(t) \lambda (dt) - \iint \pi(s) \cdot y(t) \nu (ds\ dt)
\]
for any measurable function $y$. But then note that $\gamma_n$ is constant for each $n$; let this constant measure be denoted $\gamma$. Since $\Vert \gamma_n \Vert \to 0$, this implies $\Vert \gamma \Vert =0$, and that $\gamma \cdot y = 0$ for any such $y$. 

Now fix $A\subseteq T$ and let $y$ be given by
\[
y(t) = \left\lbrace 
\begin{array}{lr}
0 & \text{ if } t\not\in A\\
{\bf 1}(S) & \text{ if } t\in A
\end{array}
\right.
\]
where ${\bf 1}(S)$ is the indicator of $S$. Then 
\begin{eqnarray*}
\gamma\cdot y &=& \int_A \pi(t) \cdot y(t) \lambda (dt) - \iint_{T\times A} \pi(s) \cdot y(t) \nu(ds\ dt)\\
&=& \int_A \lambda(dt) - \iint_{T\times A} \nu(ds\ dt)\\
&=& \lambda(A) - \nu(T\times A)
\end{eqnarray*}
And $\gamma \cdot y = 0$, which implies $\lambda(A) - \nu(T\times A)= 0$, that is, $\lambda(A) = \nu(T\times A)$. Since $A$ was arbitrary, 
$\lambda(A) = \nu(T\times A)$ for each $A\subseteq T$. From this it follows first that $\lambda(T) = \nu(T\times T)$, and since $\nu>0$, this implies $\lambda(T) = \nu(T\times T) >0$. Then without loss of generality, rescaling if necessary, take $\lambda(T) = \nu(T\times T) = 1$. Second, this implies that, using disintegration of measures, we can write
\[
\nu = \int \nu_t(ds) \lambda(dt)
\]
where $\nu_t$ is a measure on $T$, $\nu_t \geq 0$ and $\nu_t(T) = 1$ for each $t$ in the support of $\lambda$. 

For  each  $t\in T$, let 
\[
\gamma(t) = \pi(t) - \int \pi(s) \nu_t(ds)
\]
Then $\gamma$ is the measure given by
\[
\gamma(E) = \int \gamma(t) (E_t) \lambda(dt) \ \ \ \ \ \forall E\subseteq T\times S
\]
where for $E\subseteq T\times S$, $E_t:= \{ r\in S: (t,r) \in E\}$. 

Then note that 
\begin{eqnarray*}
\Vert \gamma \Vert &=& \sup_E \Vert \gamma(E) \Vert \ \ \ \ \mbox{ by definition}\\
&=& \gamma^+(T\times S) + \gamma^-(T\times S) \ \ \ \ \mbox{ by definition }\\
&=& \int \left\lbrack \gamma^+(t) (S) + \gamma^-(t) (S) \right\rbrack \lambda(dt)\\
&=& \int \Vert \gamma(t)\Vert \lambda (dt) 
\end{eqnarray*}
Recall from above 
\[
\Vert \gamma \Vert =  \int \Vert \gamma(t)\Vert \lambda (dt) =  0
\]
By definition, $\Vert \gamma(t)\Vert \geq 0$ for each $t\in T$, hence $\Vert \gamma(t)\Vert = 0$ for $\lambda-\mbox{a.e } t\in T$. 

Thus for $\lambda-\mbox{a.e } t\in T$, 
\[
\gamma(t)  =   \pi(t) - \int \pi(s) \nu_t(ds) =0
\]
where $\nu_t \in \Delta(T)$. Thus by probabilistic independence, $\nu_t = \delta_t$ for $\lambda-\mbox{a.e } t\in T$. 

But then
\begin{eqnarray*}
\nu\cdot d &=& \iint d(s,t) \nu(ds\ dt) \\
&=& \iint d(s,t) \nu_t(ds) \lambda(dt)\\
&=& \int_{\mbox{supp } \lambda } d(t,t) \lambda (dt)\\
&=& 0  \ \ \ \ \mbox{ since } d(t,t) = 0 \mbox{ for all } t\in T
\end{eqnarray*}
This is a contradiction, as $\nu\cdot d >0$. Thus $p^* \leq 0$. 
\end{proof}
\bigskip
\bigskip

\section{Discussion}

The techniques developed in sections 4 and 5 provide new insights into the foundational result of McAfee and Reny (1992), and the surplus extraction problem more generally. Surplus extraction theorems are central results in mechanism design. Their conclusions that designers can extract all, or virtually all, information rents under standard assumptions yield strong and important predictions. These predictions are arguably implausible in a variety of practical settings, motivating significant work reconsidering many standard assumptions in mechanism design. New methods that help to understand what drives these results and their limitations are useful as a consequence. The constructive methods developed in section 4 highlight the connection between surplus extraction and the geometry of the set of beliefs, while the methods developed in section 5 instead highlight the connection between surplus extraction and duality by using the natural characterization of full surplus extraction contracts as solutions to families of inequalities.  Both techniques are useful for understanding when surplus extraction holds, when it might fail, and limits on designers in settings where virtual or full extraction fails. Both techniques also suggest methods for studying surplus extraction in models beyond the standard case.

To indicate how the methods we developed can be used to study surplus extraction in other settings, we close by considering a modification of the standard model in which each type is associated with a set of beliefs rather than a single belief. This set could stem from the designer's objective to have mechanisms that are robust to misspecifying agents' beliefs, or from agents' perceptions of ambiguity. We use the constructive methods from section 4 in a simple example in this framework, building on Example 1, to illustrate.

\begin{figure}
\centering
\begin{tikzpicture}[scale = .6]

\fill (5,2) circle (1pt);

\fill (5,0) circle (1pt);

\draw   (5,2) .. controls (5,5) and (1,5) .. (1,1);

\draw  (1,1) .. controls (1,-3) and (5,-3) .. (5,0);

\fill (5,2) circle (2pt);

\fill (5,0) circle (2pt);

\draw (4,1) -- (5,2);

\draw (4,-1) -- (5,0);

\coordinate [label=right: $\pi_2(0)$] (a) at (5,2);

\coordinate [label=right: $\pi_2(1)$] (b) at (5,0);

\coordinate [label=left: $\pi_1(0)$] (a) at (4,1);

\coordinate [label=left: $\pi_1(1)$] (b) at (4,-1);

\fill (4,1) circle (1pt);

\fill (4,-1) circle (1pt);

\draw  (4,1) .. controls (4,4) and (0,4) .. (0,0);

\draw  (0,0) .. controls (0,-4) and (4,-4) .. (4,-1);

\fill (4,1) circle (2pt);

\fill (4,-1) circle (2pt);

\draw (3, 2.95) -- (4, 3.95);

\fill (3,2.95) circle (2pt);

\fill (4,3.95) circle (2pt);

\coordinate [label=above: $\pi_2(t)$] (b) at (4,3.95);

\coordinate [label=below left: $\pi_1(t)$] (a) at (3,2.95);

\end{tikzpicture}

\caption{Graph of $\{ \Pi(t): t\in T\}$}
\end{figure}

\noindent {\bf Example 2: } Suppose $S$ is finite. Let $\pi_1, \pi_2:T\to \Delta(S)$ where $\pi_1$ is as in Example 1, and $\pi_2$ is a translation of $\pi_1$ as depicted in Figure 8.
For each $t\in T$, set $\Pi(t) = \{ a \pi_1(t) + (1-a) \pi_2(t) : a\in [0,1]\}$. See Figure 8. Then $\Pi(t) \subseteq \Delta(S)$ is compact and convex for each $t\in T$, and $\Pi:T\to 2^{\Delta(S)}$ is a continuous correspondence.  

Given $v:T\to {\bf R}$, for each $t\in T$ define $V_t:{\bf R}^S \to {\bf R}$ by 
\[
V_t(x) = \min_{\pi \in \Pi(t)} \pi \cdot (v(t) - x) = v(t) - \max_{\pi \in \Pi(t)} \pi\cdot x
\]
For each $t\in T$, $V_t(x)$ gives the minimum expected surplus for type $t$ from the contract $x\in {\bf R}^S$, computed with respect to beliefs in $\Pi(t)$. This is a version of maxmin expected utility, as in Gilboa and Schmeidler (1989). 

We consider the designer's surplus extraction problem, given these utilities for agents. As we noted above, this can be motivated either by robustness concerns of the designer, or the perception of ambiguity by agents. {\it Virtual extraction} holds here if for each $v:T\to {\bf R}$ and for each $\varepsilon >0$, there exists a menu $\{ c(t) \in {\bf R}^S : t\in T\}$ such that for each type $t\in T$,
\[
0\leq V_t(c(t)) \leq \varepsilon \ \ \mbox{ and } \ \ V_t(c(s)) \leq \varepsilon \ \ \forall s\not= t
\]
We adapt the constructive argument in section 4 to show that virtual extraction holds in this setting.\footnote{As in the standard case, it is straightforward to show that virtual extraction can also be achieved with a finite menu; we omit this extension of Theorem 2.}

First, let $t \in T\setminus  \{0,1\}$. In this case, there exists $z(t) \in {\bf R}^S$ such that 
\[
\pi \cdot z(t) = 0 \ \ \ \forall \pi \in \Pi(t) \ \mbox{ and } \ \pi \cdot z(t) >0 \ \ \ \forall \pi\in \Pi(s), \ \forall s\not= t
\]
See Figure 9. As in Example 1, we construct a contract of the form $c(t) = v(t) + \alpha(t) z(t)$ for appropriate choice of $\alpha(t)$. 

Let $\varepsilon >0$ be given. Choose $\delta >0$ such that $\Vert s-t\Vert < \delta \Rightarrow v(s) - v(t) < \varepsilon$. Then choose $\alpha(t) >0$ sufficiently large so that
\[
\alpha (t) > \max_{\substack{\Vert s-t\Vert \geq \delta \\ \pi\in \Pi(s)} } \frac{v(s) - v(t)}{\pi \cdot z(t)}
\]
For the resulting contract $c(t) = v(t) + \alpha(t) z(t) $, 
\[
V_t(c(t)) = v(t) - \max_{\pi\in \Pi(t)} \pi \cdot c(t) = v(t) - v(t) - \alpha(t) \max_{\pi\in \Pi(t)} \pi\cdot z(t) =0
\] 
and for all $s\not= t$, 
\[
V_s(c(t)) = v(s) - \max_{\pi\in \Pi(s)} \pi \cdot c(t) = v(s) - v(t) - \alpha(t) \max_{\pi\in \Pi(s)}\pi\cdot z(t) \leq \varepsilon 
\]

\begin{figure}
\centering
\begin{tikzpicture}[scale = .6]

\draw [gray!20, fill=gray!20] (1,3.95) -- (4.5,2.25) -- (5.5,3.25) -- (2,4.95) ;

\fill (5,2) circle (1pt);

\fill (5,0) circle (1pt);

\draw   (5,2) .. controls (5,5) and (1,5) .. (1,1);

\draw  (1,1) .. controls (1,-3) and (5,-3) .. (5,0);

\fill (5,2) circle (2pt);

\fill (5,0) circle (2pt);

\draw (4,1) -- (5,2);

\draw (4,-1) -- (5,0);

\coordinate [label=right: $\pi_2(0)$] (a) at (5,2);

\coordinate [label=right: $\pi_2(1)$] (b) at (5,0);

\coordinate [label=left: $\pi_1(0)$] (a) at (4,1);

\coordinate [label=left: $\pi_1(1)$] (b) at (4,-1);

\coordinate [label = left: $z(t)$] (c) at (1.25,3.25);

\fill (4,1) circle (1pt);

\fill (4,-1) circle (1pt);

\draw  (4,1) .. controls (4,4) and (0,4) .. (0,0);

\draw  (0,0) .. controls (0,-4) and (4,-4) .. (4,-1);

\fill (4,1) circle (2pt);

\fill (4,-1) circle (2pt);

\draw (3, 2.95) -- (4, 3.95);

\fill (3,2.95) circle (2pt);

\fill (4,3.95) circle (2pt);

\draw (2,4.95) -- (5.5, 3.25);

\draw (1,3.95) -- (4.5, 2.25);

\draw [->] (1.5, 3.75) -- (1.25, 3.25);

\coordinate [label=above: $\pi_2(t)$] (b) at (4,3.95);

\coordinate [label=below left: $\pi_1(t)$] (a) at (3,2.95);

\end{tikzpicture}

\caption{Construction of $c(t)$, $t\in T\setminus \{ 0,1\}$.}
\end{figure}

Now consider $t\in \{ 0,1\}$, and without loss of generality take $t=0$. First note that, as depicted in the left side of Figure 10, there exists $z(0) \in {\bf R}^S$ such that
\[
\pi \cdot z(0) =  0 \ \ \ \forall \pi\in \Pi(t), \ t= 0,1  \ \mbox{ and } \ \pi \cdot z(0) >0 \ \ \ \forall \pi\in \Pi(s), \ \forall s\not\in\{ 0,1\}
\]
Then note that there exists $z_1(0) \in {\bf R}^S$ such that 
\[
\pi \cdot z_1(0) = 0  \ \ \ \forall \pi\in \Pi(0) \ \mbox{ and } \ \pi \cdot z_1(0) > 0 \ \ \ \forall \pi \in \Pi(1)
\]
See the right side of Figure 10. Choose $\alpha_1(0) > 0$ sufficiently large so that 
\[
\alpha_1(0) > \max_{\pi\in \Pi(1)} \frac{v(1) - v(0) } { \pi \cdot z_1(0)}
\]
Then set $c_1(0) = v(0) + \alpha_1(0) z_1(0)$. By construction, 
\[
V_0(c_1(0)) = v(0) - \max_{\pi\in \Pi(0)} \pi \cdot c_1(0) = v(0) -v(0) -\alpha_1(0) \max_{\pi \in \Pi(0)} \pi \cdot z_1(0) = 0
\]
and
\[
V_1(c_1(0)) = v(1) -\max_{\pi\in \Pi(1)} \pi \cdot c_1(0) = v(1) -v(0) -\alpha_1(0) \max_{\pi\in \Pi(1)} \pi \cdot z_1(0) <0
\]
The contract $c_1(0)$ might leave surplus for other types $t\not\in \{0,1\}$, however. As in Example 1, we use the additional stochastic payment $z(0)$ to account for this. 

\begin{figure}
\centering
\begin{tikzpicture}[scale = 0.6]

\fill (5,2) circle (1pt);

\fill (5,0) circle (1pt);

\draw [gray!20, fill=gray!20] (4,2.5) -- (4,-3.5) -- (5,-2.5) -- (5,3.5) ;

\draw   (5,2) .. controls (5,5) and (1,5) .. (1,1);

\draw  (1,1) .. controls (1,-3) and (5,-3) .. (5,0);

\fill (5,2) circle (2pt);

\fill (5,0) circle (2pt);

\draw (4,1) -- (5,2);

\draw (4,-1) -- (5,0);

\coordinate [label=right: $\pi_2(0)$] (a) at (5,2);

\coordinate [label=right: $\pi_2(1)$] (b) at (5,0);

\coordinate [label=left: $\pi_1(0)$] (a) at (4,1);

\coordinate [label=left: $\pi_1(1)$] (b) at (4,-1);

\fill (4,1) circle (1pt);

\fill (4,-1) circle (1pt);

\draw  (4,1) .. controls (4,4) and (0,4) .. (0,0);

\draw  (0,0) .. controls (0,-4) and (4,-4) .. (4,-1);

\draw  (4,2.5) -- (4,-3.5);

\draw  (5,3.5) -- (5,-2.5);

\fill (4,1) circle (2pt);

\fill (4,-1) circle (2pt);

\draw (3, 2.95) -- (4, 3.95);

\fill (3,2.95) circle (2pt);

\fill (4,3.95) circle (2pt);

\coordinate [label=above: $\pi_2(t)$] (b) at (4,3.95);

\coordinate [label=below left: $\pi_1(t)$] (a) at (3,2.95);

\coordinate [label=below: $z(0)$] (c) at (3.25,-3);

\draw [->] (4,-3) -- (3.5,-3);

\coordinate [label=left: $\pi_1(0)$] (a) at (13,1);

\coordinate [label=left: $\pi_1(1)$] (b) at (13,-1);

\coordinate [label = left: $z_1(0)$] (c) at (11.75,2.25);

\draw [gray!20, fill=gray!20] (12,3) -- (13,1) -- (13.50,0) -- (14.50,1) -- (14,2) -- (13,4);

\draw [->] (12.25, 2.50) -- (11.75, 2.25);

\draw (13,1) -- (13,-1);

\draw (14,2) -- (14,0);

\draw (12,3) -- (13, 1) -- (13.50,0);

\draw (13,4) -- (14,2) -- (14.50,1);

\fill (13,1) circle (2pt);

\fill (13,-1) circle (2pt);

\fill (14,2) circle (2pt);

\fill (14,0) circle (2pt);

\draw (13,1) -- (14,2);

\draw (13,-1) -- (14,0);

\coordinate [label=right: $\pi_2(0)$] (a) at (14,2);

\coordinate [label=right: $\pi_2(1)$] (b) at (14,0);

\end{tikzpicture}

\caption{Construction of $c_1(0)$ and $c(0)$. }
\end{figure}

To that end, for each $t\in T$ set 
\[
\overline{V}_t(c_1(0)) = v(t) - v(0) - \min_{\pi\in \Pi(t)} \pi\cdot (\alpha_1(0)z_1(0))
\]
For type $t$, $\overline{V}_t(c_1(0))$ is the maximum expected surplus from the contract $c_1(0)$ over all beliefs in $\Pi(t)$. 
Note that by construction, $\overline{V}_0(c_1(0)) = 0$ and $\overline{V}_1(c_1(0))<0$. Then choose $\delta >0$ such that for $\Vert s-0\Vert < \delta$ or $\Vert s - 1\Vert <\delta$, $\overline{V}_s(c_1(0))  <\varepsilon$. This is possible, by the continuity of $v$ and $\Pi$, and the fact that  $\overline{V}_0(c_1(0)) = 0$ and $\overline{V}_1(c_1(0))<0$.  

Now set $\alpha(0) >0$ such that
\[
\alpha(0) > \max_{\substack{\Vert s-0\Vert \geq \delta \\ \Vert s-1\Vert \geq \delta \\ \pi\in \Pi(s)}} \frac{ \overline{V}_s(c_1(0))}{\pi \cdot z(0)}
\]
and set 
\[
c(0)= c_1(0) + \alpha(0) z(0) = v(0) + \alpha_1(0) z_1(0) + \alpha(0) z(0)
\]
By construction, for types $t \in \{0,1\}$, $\pi \cdot c(0) = \pi \cdot c_1(0)$  for all $\pi \in \Pi(t)$. Thus 
$V_0(c(0))= 0$ and $V_1(c(0)) <0$. 
For $t\not\in \{ 0,1\}$, again by construction, 
\[
V_t(c(0)) = v(t) - v(0) - \max_{\pi\in \Pi(t)} \pi \cdot (\alpha_1(0) z_1(0) +\alpha(0)  z(0) ) < \varepsilon
\]
The collection $\{ c(t) : t\in T\}$ thus constructed achieves extraction of all but at most $\varepsilon$ surplus. \hfill$\diamondsuit$

\section{Appendix}

Because these results might be of independent interest, we include here the derivation of the version of Ekeland's Variational Principle that we used in the proof of Theorem 7.

Before giving the main result of the appendix, we start with some preliminary notation, definitions, and results, including the classic version of Ekeland's Variational Principle, and an extension due to Borwein, from which the main result follows quickly. 

For a set $X$ and an extended real-valued function $f:X\to {\bf R} \cup \{ +\infty \}$, the \emph{effective domain} of $f$, denoted $\mbox{ dom } f$, is the set of points $x \in X$ such that $f(x) \in {\bf R}$. An extended real-valued function $f$ is \emph{proper} if $\mbox{ dom } f \not= \emptyset$. Let $\inf f = \inf_{x\in X} f(x)$ below. If $X$ is a topological vector space, $X^*$ denotes its dual, and for $x\in X$ and $x^*\in X^*$, $\langle x^*, x\rangle = x^*(x)$. 

\begin{definition}
Let $X$ be a topological vector space and $f:X\to {\bf R} \cup \{ +\infty\}$. For $x\in \mbox{ dom } f$ and $\varepsilon \geq 0$, the \emph{$\varepsilon$-subdifferential} of $f$ at $x$, denoted $\partial_\varepsilon f(x)$, is 
\[
\partial_\varepsilon f(x) = \{ x^*\in X^*: f(y) \geq f(x) + \langle x^*, y-x\rangle -\varepsilon \ \ \forall y\in X\}
\]
\end{definition}

\noindent {\bf Note: } For $\varepsilon = 0$, $\partial_0 f(x) = \partial f(x)$, the standard subdifferential. For any $x\in \mbox{ dom } f$ and any $\varepsilon >0$, $\inf f \leq f(x) \leq \inf f + \varepsilon \iff 0\in \partial_\varepsilon f(x)$. 

Next we state the classic version of Ekeland's Variational Principle (Ekeland, 1974). 

\begin{theorem}
{\rm \bf (Ekeland's Variational Principle)} Let $(V,d)$ be a complete metric space and $F:V\to {\bf R}\cup \{ +\infty\}$ be a proper, lower semicontinuous function such that $\inf F > -\infty$. Let $\varepsilon >0$ and $\beta >0$. For every $u\in V$ such that 
\[
\inf F \leq F(u) \leq \inf F + \varepsilon
\]
there exists $v\in V$ such that 
\begin{itemize}
  \item[(i)] $F(v) \leq F(u)$
  \item[(ii)] $d(u,v) \leq \beta$
  \item[(iii)] $F(u) \geq F(v) - \frac{\varepsilon}{\beta} d(v,w) \ \ \ \forall w\not= v$   
\end{itemize}
If in addition $F$ is convex, then 
\begin{itemize}
  \item[(iv)] $v$ can be chosen such that there exists $g\in \partial F(v)$ such that $\Vert g \Vert \leq \frac{\varepsilon}{\beta}$  
\end{itemize}
\end{theorem}

More precise approximations can be given for convex functions, as shown by Borwein (1982).  

\begin{theorem}
{\rm \bf (Borwein, 1982, Theorem 1)} Let $X$ be a Banach space and  $f:X\to {\bf R} \cup \{ +\infty\}$ be a proper, convex, lower semicontinuous function. Let $\varepsilon >0$ and $k \geq 0$ be given. Let 
\[ 
x_0^*\in \partial_\varepsilon f(x_0)
\]
Then there exist $x_\varepsilon$ and $x_\varepsilon^*$ such that
\[
x_\varepsilon^* \in \partial f(x_\varepsilon)
\]
and such that
\begin{eqnarray*}
\Vert x_\varepsilon - x_0\Vert &\leq & \sqrt{\varepsilon}\\
\Vert f(x_\varepsilon) - f(x_0) \Vert &\leq& \sqrt{\varepsilon} (\sqrt{\varepsilon} + \frac{1}{k})\\
\Vert x_\varepsilon^* - x_0^* \Vert &\leq & \sqrt{\varepsilon} (1+ k\Vert x_0^*\Vert )\\
\Vert x_\varepsilon^*(h) - x_0^*(h)\Vert &\leq & \sqrt{\varepsilon} ( \Vert h \Vert + k\Vert x_0^*(h)\Vert)\\
x_\varepsilon^* &\in& \partial_{2\varepsilon}f(x_0)
\end{eqnarray*}
\end{theorem}

Putting these two results together yields the following. 

\begin{lemma}
 Let $X$ be a Banach space and $f:X \to {\bf R} \cup \{ +\infty \}$ be a proper, convex, lower semicontinuous function such that $\inf f > -\infty$. Then there exists a sequence $\{ x_n\}$ in $X$ such that $f(x_n) \to \inf f$ and $d(0, \partial f(x_n)) \to 0$, i.e., there exists $\{ g_n\}$ such that $g_n \in \partial f(x_n)$ for each $n$ and $\Vert g_n \Vert \to 0$. 
\end{lemma}

\begin{proof}
For each $n\in {\bf N}$ there exists $x_n \in X$ such that 
\[
\inf f \leq f(x_n) \leq \inf f + \frac{1}{4n^2}
\]
Then $0\in \partial_{\frac{1}{4n^2}} f(x_n)$ for each $n$. 

By Borwein (1982, Theorem 1), for each $n$ there exist $\bar x_n$ and $\bar x_n^*$ such that (with $k=1$ here)
\begin{eqnarray*}
\bar x_n^* &\in & \partial f(\bar x_n)\\
\Vert \bar x_n - x_n \Vert &\leq& \frac{1}{2n}\\
\Vert f(\bar x_n) - f(x_n)\Vert &\leq & \frac{1}{2n} (\frac{1}{2n} + 1)\\
&\leq & \frac{1}{4n^2} + \frac{1}{2n} \leq \frac{1}{n}\\
\Vert \bar x_n^* - 0\Vert = \Vert \bar x_n^* \Vert &\leq & \frac{1}{2n}(1+ 0) = \frac{1}{2n}\\
\Vert \bar x_n^*(h)\Vert &\leq & \frac{1}{2n} \Vert h\Vert
\end{eqnarray*}
Then for each $n$,
\begin{eqnarray*}
\inf f \leq f(\bar x_n) &\leq& f(x_n) + \frac{1}{n}\\
&\leq& \inf f + \frac{1}{4n^2}+ \frac{1}{n}\\
&\leq & \inf f + \frac{2}{n}
\end{eqnarray*}
and 
\[
\bar x_n^* \in \partial f(\bar x_n) \mbox{ with } \Vert \bar x_n^* \Vert \leq \frac{1}{2n}
\]
Thus $f(\bar x_n) \to \inf f$ and $\Vert \bar x_n^*\Vert \to 0$. 
\end{proof}



\begin{lemma}
Let $B$ be a compact metric space and $X= C(B)$. Let $f:B\times X \to {\bf R}$ be continuous, and  for each $b\in B$, let $f_b:X\to {\bf R}$ be given by $f_b(x) = f(b,x)$. Suppose $f_b$ is convex for each $b\in B$. Let $h:X\to {\bf R}$ be given by 
\[
h(x) = \int f_b(x) \mu(db)
\]
where $\mu \in {\cal M}(B)$ and $\mu\geq 0$. Then $h$ is convex, and 
\[
\partial h(x) = \int \partial f_b(x) \mu(db)
\]
That is, for each $\gamma \in \partial h(x)$ there is a mapping $b\mapsto \gamma_b$ such that $\gamma_b \in \partial f_b(x)$ for $\mu$-$\mbox{ a.e } b\in B$ and 
\[
\gamma\cdot y = \int \gamma_b \cdot y \ \mu(db)
\]
for any measurable $y$. 
\end{lemma}
\begin{proof}
Since $B$ is a compact metric space, $X=C(B)$ is separable. The result then follows from Ioffe and Levin (1972); see also Clarke (1990) Theorem 2.7.2 and discussion on pp. 76-77.
\end{proof}

\begin{lemma}
Let $f:T\times C(T\times S) \to {\bf R}$ be given by 
\[
f(t,z) =  \pi(t) \cdot z(t)
\]
and for each $t\in T$, let $f_t:C(T\times S) \to {\bf R}$ be given by $f_t(z) = f(t,z)$. 
Then $f$ is continuous and $f_t$ is convex for each $t\in T$. For each $t\in T$,  if 
$\gamma \in \partial f_t(z)$ then $\gamma \in {\cal M}(T\times S)$ is a measure such that for any measurable function $y$,
\[
 \gamma \cdot y = \pi(t) \cdot y(t) 
\]
\end{lemma}
\begin{proof}
First, $f$ is continuous, by construction. 
To see this, let $(s_n,z_n)\to (t,z)$. Then 
\begin{eqnarray*}
\Vert \pi(s_n)\cdot z_n(s_n)-\pi(t)\cdot z(t) \Vert &=& \Vert (\pi(s_n) - \pi(t) ) \cdot z_n(s_n) + \pi(t)\cdot (z_n(s_n) -z(t) ) \Vert\\
&\leq& \Vert  (\pi(s_n) - \pi(t) ) \cdot z_n(s_n)\Vert + \Vert \pi(t)\cdot (z_n(s_n) -z(t) ) \Vert\\
&\leq& \Vert \pi(s_n) - \pi(t) \Vert \Vert z_n(s_n)\Vert + \Vert \pi(t)\cdot (z_n(s_n) -z(t) ) \Vert\\
\end{eqnarray*}
Since $z_n\to z$, $\{ z_n(s_n)\} $ and $\{ z_n(s_n) - z(t) \}$ are bounded, and $z_n(s_n)\to z(t)$ pointwise. Then $\Vert \pi(s_n) - \pi(t) \Vert \Vert z_n(s_n)\Vert\to 0$, since $\pi(s_n)\to \pi(t)$ in norm, and $\pi(t)\cdot (z_n(s_n) -z(t) ) \to 0$ by the bounded convergence theorem. Thus $f(s_n, z_n) \to f(t, z)$. 

By construction, $f_t$ is linear and continuous for each $t\in T$, and 
$\partial f_t(z)$ is the measure $\gamma_t\in {\cal M}(T\times S)$ such that 
\[
\gamma_t\cdot y = \pi(t)\cdot y(t) \ \ \mbox{ for $y$ measurable function }
\]
\end{proof}

\begin{lemma}
Let $g:T\times T\times C(T\times S)\to {\bf R}$ be given by 
\[
g_{(s,t)}(z) := g(s,t,z) =  \pi(s) \cdot z(t)
\]
Then $g$ is continuous, and $g_{(s,t)}$ is concave for each $t,s\in T$. For each $t,s\in T$, if $\gamma \in \partial g_{(s,t)}(z)$ then $\gamma \in {\cal M}(T\times S)$ is a measure such that for any measurable function $y$,
\[
 \gamma \cdot y = \pi(s) \cdot y(t) 
\]
\end{lemma}
\begin{proof}
This follows from arguments analogous to those used in the proof of Lemma 6. 
\end{proof}

\begin{lemma}
Let $\lambda \in {\cal M}(T)$ and $\nu \in {\cal M}(T\times T)$. Let $d\in C(T\times T)$ be given by $d(t,s) = v(t) - v(s)$. Let $F:C(T\times S) \to {\bf R}$ be given by 
\[
F(z) := \lambda\cdot f(z) + \nu \cdot(d-g(z))
\]
Then $F$ is convex and continuous. If $\gamma\in \partial F(z)$, then $\gamma \in {\cal M}(T\times S)$ is a measure for which 
\[
\gamma\cdot y = \int \pi(t) \cdot y(t) \lambda (dt) - \iint \pi(s) \cdot y(t) \nu (ds\ dt)
\]
for any measurable function $y$. 
\end{lemma}
\begin{proof}
This follows from Lemmas 5, 6, and 7. 
\end{proof}

\begin{lemma}
Let $T$ be a compact metric space. For each $v:T\to {\bf R}$, if $c\in C(T\times S)$ satisfies weak full extraction, then there exists $c' \in C(T\times S)$ satisfying full extraction. 
\end{lemma}
\begin{proof}
Fix $v:T\to {\bf R}$ and suppose $c\in C(T\times S)$ satisfies weak full extraction. By assumption, for each $t\in T$
\begin{eqnarray*}
v(t) - \pi(t) \cdot c(t) &\geq& 0 \\
v(t) - \pi(t)\cdot c(s) &\leq & 0 \ \ \ \forall s\not= t
\end{eqnarray*}
For each $t\in T$, set 
\[
w(t):= v(t) - \pi(t)\cdot c(t)
\]
Then $w:T\to {\bf R}$ is continuous, and $w(t) \geq 0$ for each $t$. Now for each $t\in T$ set
\[
c'(t) = c(t) + w(t) {\bf 1}(S)
\]
Then $c'\in C(T\times S)$, and for each $t\in T$,
\[
v(t) - \pi(t)\cdot c'(t) = 0
\]
while for $s\not= t$, 
\[
v(t)-\pi(t)\cdot c'(s) = v(t) - \pi(t) \cdot c(s) - w(s) \leq v(t) -\pi(t)\cdot c(s) \leq 0
\]
Thus $c'$ satisfies full extraction.
\end{proof}

\noindent {\bf Proof of Theorem 4:} 
Let $x \in C$ be an extreme point. If $x$ is an exposed point of $C$, then it is also eventually exposed. So suppose $x$ is not exposed. 
Since $x$ is an extreme point of $C$, there exists $z_1\in {\bf R}^k$ and $b_1\in {\bf R}$ such that
\begin{eqnarray*}
x\cdot z_1 &=& b_1\\
y\cdot z_1 &\geq & b_1 \ \forall y\in C\\
y\cdot z_1 &>& b_1 \ \mbox{ for some } y\in C
\end{eqnarray*}

Set 
\[
F_1 = \{ y\in C: y\cdot z_1 = b_1\}
\]
Since $x$ is not an exposed point of $C$, $F_1$ is a proper face of $C$, so $\mbox{ dim } F_1 < \mbox{ dim } C$,  and $\mbox{ dim } F_1 \geq 1$. 

If $x$ is an exposed point of $F_1$, we are done, setting $F_2 = \{x\}$. Else, $x$ must be an extreme point of $F_1$ (since $F_1 \subseteq C$) that is not exposed. Repeating the above argument, choose $z_2 \in {\bf R}^k$ and $b_2\in {\bf R}$ such that 
\begin{eqnarray*}
x\cdot z_2 &=& b_2\\
y\cdot z_2 &\geq & b_2 \ \forall y\in F_1\\
y\cdot z_2 &>& b_2 \ \mbox{ for some } y\in F_1
\end{eqnarray*}
Set 
\[
F_2 = \{ y\in F_1: y\cdot z_2 = b_2\}
\]
If $x$ is an exposed point of $F_2$, we are done. Else, $F_2$ is a proper face of $F_1$ and
\[
1 \leq \mbox{ dim } F_2 < \mbox{ dim } F_1 \leq \mbox{ dim } C -1
\]
Repeating this argument, since $\mbox{ dim } F_i \leq \mbox{ dim } F_{i-1} -1$ for each $i$, eventually must have $\mbox{ dim } F_n = 1$, and then because $x$ is an extreme point of $F_n$, $x$ must also be an exposed point of $F_n$. \qed

\end{document}